\documentclass[a4paper,12pt]{article}

\usepackage{amsmath,amssymb,mathrsfs,amsthm,bm,color}
\usepackage{enumerate}

\newtheorem{theorem}{\bf Theorem}[section]
\newtheorem{lemma}[theorem]{\bf Lemma}
\newtheorem{proposition}[theorem]{\bf Proposition}
\newtheorem{corollary}[theorem]{\bf Corollary}

\newtheorem{remark}{\bf Remark}[section]
\newtheorem{assumption}{\bf Assumption}[section]

\makeatletter
 
  \@addtoreset{equation}{section}
 \makeatother

\title{An inverse scattering problem for \\the Klein-Gordon equation 
	with \\ a classical source in quantum field theory}
\author{Hironobu Sasaki
\thanks{Department of Mathematics and Informatics, 
    Chiba 263-8522, Japan, 
    E-mail: \texttt{sasaki@math.s.chiba-u.ac.jp}} 
\\ and \\
Akito Suzuki
\thanks{Department of Engineering, Shinshu University, 
	Nagano, 380-8553, Japan
    E-mail: \texttt{akito@shinshu-u.ac.jp}}
} 

\begin{document}

\maketitle
\begin{center}
\begin{minipage}{115mm}
{\footnotesize 

}
\end{minipage}
\end{center}
\abstract{An inverse scattering problem for a quantized scalar field ${\bm \phi}$
obeying a linear Klein-Gordon equation
\[ (\square + m^2 + V) {\bm \phi} = J 
\quad \mbox{in $\mathbb{R} \times \mathbb{R}^3$} \]
is considered, 
where $V$ is a repulsive external potential
and $J$ an external source $J$.
We prove that the scattering operator 
$\mathscr{S}= \mathscr{S}(V,J)$ associated with ${\bm \phi}$
uniquely determines $V$.
Assuming that $J$ is of the form 
$J(t,x)=j(t)\rho(x)$, $(t,x) \in \mathbb{R} \times \mathbb{R}^3$,
we represent $\rho$ (resp. $j$) 
in terms of $j$ (resp. $\rho$) and $\mathscr{S}$. }

\section{Introduction}
We consider 
an inverse scattering problem for a quantized scalar field ${\bm \phi}$
interacting with an external potential $V$ and an external source $J$
(see, e.g., \cite{HT, PS}) which obeys the Klein-Gordon equation
	\begin{align}
	\label{KG}
		(\square + m^2 + V(x) ){\bm \phi}(t,x) & = J(t,x)
		\quad \mbox{in $\mathbb{R} \times \mathbb{R}^3$}.
	\end{align} 
Here $\square = \frac{\partial^2}{\partial t^2}-\Delta$, 
$\Delta$ is the Laplacian in $\mathbb{R}^3$,
$m>0$ and $J$, $V$ are real functions.
${\bm \phi}(t,x)$ and its conjugate field ${\bm \pi}(t, x) = \frac{\partial}{\partial t}{\bm \phi}(t,x)$ 
are operator valued distributions (see, e.g., \cite{WightmanGarding}).
A typical example of \eqref{KG} is the nucleon-pion interaction,
that is, ${\bm \phi}$ describes the pion field and $J$ the distribution function of the nucleons (see, e.g., \cite{HT}).

Under suitable conditions, 
one can show that the asymptotic fields 
\begin{align}
\label{phipm}
& {\bm \phi}_{\rm out/in}(t,x) 
= \mbox{s-} \lim_{s \to \pm \infty}{\bm \phi}_s(t,x), \\
\label{pipm}
& {\bm \pi}_{\rm out/in}(t,x) 
= \mbox{s-} \lim_{s \to \pm \infty}{\bm \pi}_s(t,x)
\end{align}
exist. 
Here 
${\bm \pi}_s(t,x) = \frac{\partial}{\partial t}{\bm \phi}_s(t,x)$
and ${\bm \phi}_s(t,x)$ is the solution
of the free Klein-Gordon equation with the initial condition:
${\bm \phi}_s(s,x)={\bm \phi}(s,x)$ and ${\bm \pi}_s(s,x)={\bm \pi}(s,x)$.
Suppose that ${\bm \phi}_{\rm in}(x) = {\bm \phi}_{\rm in}(0,x)$ 
and ${\bm \pi}_{\rm in}(x) = {\bm \pi}_{\rm in}(0,x)$
give the Fock representation of the canonical commutation relations (CCR):
\begin{align}
\label{ccr1}
		& [{\bm \phi}_{\rm in}(x),{\bm \pi}_{\rm in}(y)]=i\delta(x-y), \\ 
\label{ccr2}
		& [{\bm \phi}_{\rm in}(x), {\bm \phi}_{\rm in}(y)]
	=[{\bm \pi}_{\rm in}(x), {\bm \pi}_{\rm in}(y)]=0.
\end{align}
See Section 3 for the detail.
The scattering operator $\mathscr{S} = \mathscr{S}(V,J)$ 
is defined by the following relations (up to a constant factor): 
\begin{equation}
\label{111}
\mathscr{S}^{-1} {\bm \phi}_{\rm in}(x) \mathscr{S}
	= {\bm \phi}_{\rm out}(x), 
\quad \mathscr{S}^{-1} {\bm \pi}_{\rm in}(x) \mathscr{S}
	= {\bm \pi}_{\rm out}(x). 
\end{equation}
We prove that $\mathscr{S}$ uniquely determines $V$.
Suppose that $J(t,x) = j(t)\rho(x)$. 
Then we show that
$\rho$ (resp. $j$) is uniquely determined by $\mathscr{S}$ and $j$
(resp. $\rho$).

To state our results precisely, 
we introduce several assumptions.
We set $\mathfrak{h} = L^2(\mathbb{R}^3;dx)$
and assume the following:
\begin{assumption}
\label{asspot}
The potential function $V:\mathbb{R}^3 \to \mathbb{R}$ 
is non negative and satisfies $V \in H^2(\mathbb{R}^3)$. 
\end{assumption}
Then the multiplication operator $V$ acting in $\mathfrak{h}$ 
is infinitesimally small with respect to $h_0=-\Delta$
since $V \in L^2(\mathbb{R}^3)$.
Hence the operator 
\[ h = h_0 + V \] 
is self-adjoint with the domain $D(h)=D(h_0)=H^2(\mathbb{R}^3)$.
Since $V$ is relative compact with respect to $h_0$, 
i.e., $V(h_0+1)^{-1}$ is compact, and $V$ is positive,
the spectrum of $h$ is $\sigma(h) = \sigma_{\rm ess}(h) = [0,\infty)$.
The condition $V \in H^2(\mathbb{R}^3)$ allows us to 
construct the solution of \eqref{KG} by a Bogoliubov transformation
(see Lemma \ref{HSlemma} and Proposition \ref{repKG}).

We set 
\[ \omega = \varphi(h), \quad \omega_0 = \varphi(h_0), \]
where $\varphi(s) = \sqrt{s+m^2}$.

\begin{assumption}
\label{assscat}
We assume that $h$ has no positive eigenvalue and
\begin{equation}
\label{geom}
\int_0^\infty dR \|V(x)(-\Delta + 1)^{-1}F(|x|\geq R) \| < \infty, 
\end{equation}
where 
\[ F(|x|\geq R) 
	= \begin{cases} 1 & \mbox{if $|x| \geq R$}, \\ 0 & \mbox{if $|x| < R$}. \end{cases} \]  
\end{assumption}
We make some comments on Assumption \ref{assscat}:
\begin{itemize}
\item By \eqref{geom}, the following limits exist
\[ w_{\pm}
:= {\rm s-}\lim_{t \to \pm \infty} 
	e^{ith}e^{-ith_0} \]
and the intertwining property $h w_\pm =  w_\pm h_0$ holds.
By Enss and Weder \cite{EnssWeder}, 
we see that the scattering map  for the Schr\"odinger operator 
defined by
\[ \mathcal{V}_{\rm SR} \ni V \mapsto S(V)=w_+^*w_- \]
is injective, where 
\[ \mathcal{V}_{\rm SR} = \{ V: \mathbb{R}^3 \to \mathbb{R}
\mid \mbox{$V$ is Kato-small in $\mathfrak{h}$ and satisfies \eqref{geom}} \}. \] 
\item Since $h$ has no positive eigenvalue, 
\eqref{geom} implies that $h$ has purely absolutely continuous spectrum.
In particular, we have $w_\pm^* w_\pm = w_\pm w_\pm^* = I$ on $\mathfrak{h}$.
\item 
We see that 
the following limits exist:
\begin{align}
\label{1738}
\mbox{s-}\lim_{t \to \pm \infty} 
e^{it\omega}e^{-it\omega_0}.
\end{align} 
For a proof of the existence, 
see Appendix A.2.
By \cite[Theorem 1]{Woll},
above limits \eqref{1738} are equal to $w_\pm$, 
respectively,
i.e., 
the invariance principal holds for $\varphi$.
\end{itemize}

\begin{assumption}
\label{assext}
The function $J:\mathbb{R} \times \mathbb{R}^3 \to \mathbb{R}$
satisfies 
\begin{itemize}
\item[(a)]
For each $t \in \mathbb{R}$, the function $J_t(x):=J(t,x)$
satisfies $J_t \in H^{-1/2}(\mathbb{R}^3)$.
\item[(b)] 
The vector valued function 
$\mathbb{R} \ni t \mapsto e^{-it\omega}\omega^{-1/2}J_t \in \mathfrak{h}$
satisfies
\[ \int_{-\infty}^\infty dt 
	\|\omega^{-1/2}J_t\|_{\mathfrak{h}} < \infty. \]
\end{itemize}
\end{assumption}

We say that $V \in \mathcal{V}$ if $V$ satisfies 
Assumptions \ref{asspot} and \ref{assscat}
and that $J \in \mathcal{J}$ if $J$ satisfies
Assumption \ref{assext}. 

\begin{theorem}
\label{inj}
Suppose that $V, V^\prime \in \mathcal{V}$ and $J, J^\prime \in \mathcal{J}$.
If $\mathscr{S}(V,J) = \mathscr{S}(V^\prime,J^\prime)$,
then:
\begin{itemize}
\item[(i)] $S(V) = S(V^\prime)$,
\item[(ii)] $V = V^\prime$,
\item[(iii)] $\int_{-\infty}^{+\infty} ds e^{-is \omega}J_t
= \int_{-\infty}^{+\infty} ds e^{-is \omega}J_t^\prime$.
\end{itemize}
\end{theorem}

By the above theorem, we immediately see the following:
\begin{corollary}
Let $J \in \mathcal{J}$ be given.
Then the map $\mathcal{V} \ni V \mapsto \mathscr{S}(V,J)$ is injective.
\end{corollary}

\begin{proof}[Proof of Theorem \ref{inj}]
(i) will be proved in Theorem \ref{main1}.
Since $\mathcal{V} \subset \mathcal{V}_{\rm SR}$,
(ii) follows from the injectivity of the map $V \mapsto S(V)$.
We will prove (iii) in Proposition \ref{039}.
\end{proof}

In order to recover the external source $J$,
we henceforth suppose that $J \in \mathcal{J}$ is expressed by
\begin{align}
\label{separate}
J(t,x) = j(t) \times \rho(x),
\end{align} 
where $j \in L^1(\mathbb{R})$ and $\rho \in H^{-1/2}(\mathbb{R}^3)$.
Let
\[ F(t,f) 
= (\Omega_{\rm in}, {\bm \phi}_{\rm out}(t, f) \Omega_{\rm in}), \quad f \in \mathcal{S}(\mathbb{R}^d), \]
where $\Omega_{\rm in}$ is the Fock vacuum.
From \eqref{111},
we see that $F(t,f)$ is uniquely determined by $\mathscr{S}$.
One  can show the following:
\begin{itemize}
\item[(1)] For a given function $j$ 
such that the Fourier transform $\hat{j}$ of $j$ is real analytic,
we represent $\rho$ in terms of $j$ and $F(t,f)$.  
See Theorem \ref{Thm:3.1} for the detail.
\item[(2)] Let $\rho$ be a given function 
and assume that $j$ satisfies the following:
\begin{align}\label{condition:3.2.0}
\text{For some } \delta>0,\quad 
e^{\delta |t|} j(t)\in L^1(\mathbb{R}_t).
\end{align}
Then we express $j$ by means of $\rho$ and $F(t,f)$.
See Theorem \ref{Thm:3.2} for the detail.
\end{itemize}

From the reconstruction formulas above,
we observe that the scattering operator $\mathscr{S}(V,j\times \rho)$
uniquely determines $j$ and $\rho$: 
\begin{theorem}
\label{inj2}
Let $V \in \mathcal{V}$ and 
$j \times \rho$, $j^\prime \times \rho^\prime \in \mathcal{J}$.
Suppose that $j, j^\prime \in L^1(\mathbb{R})$ 
and $\rho, \rho^\prime \in H^{-1/2}(\mathbb{R})$.
Then:
\begin{itemize}
\item[(i)] Assume that that $\hat{j}$ is real analytic.
Then $\rho = \rho^\prime$ 
if $\mathscr{S}(V,j\times \rho) = \mathscr{S}(V, j\times \rho^\prime)$.
\item[(ii)] Assume that \eqref{condition:3.2.0} holds.
Then $j=j^\prime$
if $\mathscr{S}(V,j\times \rho) = \mathscr{S}(V, j^\prime \times \rho)$.
\end{itemize}
\end{theorem}
\begin{proof}
See Theorem \ref{Thm:3.1}. $\rho$ is uniquely determined by
$j$ and a function $z$ defined in \eqref{z}. 
$z$ is uniquely determined by $F(t,f)$.
Since, as was noted above, 
$\mathscr{S}$ determines $F(t,f)$ uniquely,
we see that $\rho \mapsto \mathscr{S}(V,j\times \rho)$ is injective.  
Hence (i) holds.
(ii) is proved similarly.
\end{proof}

This paper is organized as follows.
Section 2 is devoted to some mathematical preliminaries.
In Subsections 2.1 and 2.2,
we review well-known facts.
The quantized Klein-Gordon field is constructed in Subsection 2.3
and the wave operator in Subsection 2.4..
In Section 3, we discuss the scattering theory
and define the scattering operator $\mathscr{S}$.
Section 4 deals with the inverse scattering problem.
In Subsection 4.1, we show the uniqueness of $V$.
The reconstruction formulas of $\rho$ and $j$
are given in Subsections 4.2 and 4.3, respectively.
In Appendix, we prove Lemmas \ref{HSlemma} and \ref{waveop}.

\section{Preliminary} 

In general we denote the inner product and the associated norm of a Hilbert space $\mathcal{L}$ by 
$(*,\cdot )_{\mathcal{L}}$ and $\| \cdot \|_{\mathcal{L}}$, respectively.
The inner product is linear in $\cdot$ and antilinear in $*$.
If there is no danger of confusion, we omit the subscript $\mathcal{L}$ in 
$( \cdot,\cdot )_{\mathcal{L}}$ and $\| \cdot \|_{\mathcal{L}}$.
For a linear operator $T$ on $\mathcal{L}$,
we denote the domain of $T$ by $D(T)$ 
and, if $D(T)$ is dense in $\mathcal{H}$, the adjoint of $T$ by $T^*$.

\subsection{Boson Fock space}

We first recall the abstract Boson Fock space and operators therein.
The Boson Fock space over a Hilbert space $\mathfrak{h}$ is defined by
\begin{align*}
&\Gamma(\mathfrak{h})
	=\bigoplus_{n=0}^{\infty}\bigotimes_{\rm s}^n \mathfrak{h} \\
&= \left\{ \Psi = \{\Psi^{(n)} \}_{n=0}^{\infty} \Bigg| \Psi^{(n)} \in \bigotimes_{\rm s}^n \mathfrak{h}, 
\quad 
\sum_{n=0}^{\infty}\left\| \Psi^{(n)} \right\|^2_{\otimes_{\rm s}^n 			\mathfrak{h}} < \infty \right\}, 
\end{align*}
where $\otimes^n_{\rm s}\mathfrak{h}$ denotes the symmetric tensor product of $\mathfrak{h}$ 
with the convention $\otimes_{\rm s}^{0}\mathfrak{h}=\mathbb{C}$.  

The creation operator $c^*(f)$ ($f \in \mathfrak{h}$) 
acting in $\Gamma({\mathfrak{h}})$ is defined by
	\[ \left(c^*(f)\Psi \right)^{(n)} = \sqrt{n}S_n\left(f \otimes \Psi^{(n-1)} \right) \]
with the domain
	\[ D(c^*(f)) = \left\{ \Psi = \{\Psi^{(n)} \}_{n=0}^{\infty} \Bigg|
		\sum_{n=0}^{\infty} n \left\|S_n\left(f \otimes \Psi^{(n-1)} \right)
			\right\|^2_{\otimes_{\rm s}^n \mathfrak{h}} < \infty \right\}, \]
where $S_n$ denotes the symmetrization operator on $\otimes^n \mathfrak{h}$
satisfying $S_n = S_n^* = S_n^2$ and 
$S_n (\otimes^n \mathfrak{h})=\otimes^n_{\rm s} \mathfrak{h}$.

The annihilation operator $c(f)$ ($f \in \mathfrak{h}$) is defined 
by the adjoint of $c^*(f)$,
i.e., $c(f):=c^*(f)^*$. 
By definition, $c^*(f)$ (resp. $c(f)$) is linear (resp. antilinear) in $f \in \mathfrak{h}$. 
As is well known, the creation and annihilation operators leave
the finite particle subspace 
\[ D_{\rm f} = \bigcup_{m = 1}^{\infty}
		\left\{ \Psi = \{\Psi^{(n)} \}_{n=0}^{\infty} \mid \Psi^{(n)}=0, ~~ n \geq m \right\}\]
invariant.
The canonical commutation relations (CCR) hold on $D_{\rm f}$:
\begin{equation} 
\label{ccr}
[c(f), c^*(g)] = ( f,g )_{\mathfrak{h}}, 
	\quad [ c(f), c(g) ]=[c^*(f), c^*(g)]=0. 
\end{equation}
It follows from \eqref{ccr} that
\begin{align}
\label{913}
\|c^*(f) \Psi\|^2 = \|f\|^2 \|\Psi\|^2 + \|c(f) \Psi\|^2,
\quad \Psi \in D_{\rm f}.
\end{align}
The Segal field operator
$\tau(f) = \frac{1}{\sqrt{2}}(c(f) + c^*(f))$
($f \in \mathfrak{h}$)
is essentially self-adjoint on $D_{\rm f}$.
We denote its closure by the same symbol.
By \eqref{ccr}, the following equation holds
\begin{equation}
\label{922} 
\|c^*(f)\Psi\|^2 
= \frac{1}{2}(\|\tau(f)\Psi\|^2 + \|\tau(if) \Psi\|^2
	+ \|f\|^2\|\Psi\|^2),
\quad \Psi \in D_{\rm f}.
\end{equation}
Since $D_{\rm f}$ is a core for $c(f)$, $c^*(f)$ and $\tau(f)$ ($f \in \mathfrak{h}$),
we observe from \eqref{913} and \eqref{922} that 
\begin{equation}
\label{domR} D(\tau(f)) \cap D(\tau(if)) = D(c(f)) = D(c^*(f)). 
\end{equation} 
Hence the following operator equalities hold true:
\begin{align*}
& c(f) = \frac{1}{\sqrt{2}} (\tau(f) + i\tau(if)), \\
& c^*(f) = \frac{1}{\sqrt{2}} (\tau(f) - i\tau(if)). 
\end{align*}
Let 
\[ \tilde{D} := \bigcap_{f \in \mathfrak{h}} D(c(f)). \]
Since $\tilde{D} \supset D_{\rm f}$, $\tilde{D}$ 
is dense in $\Gamma(\mathfrak{h})$. 
From \eqref{domR}, we observe that
\begin{equation}
\label{936} 
\tilde{D} = \bigcap_{f \in \mathfrak{h}} D(c^*(f)) 
= \bigcap_{f \in \mathfrak{h}}D(\tau(f)). 
\end{equation}
The Fock vacuum 
$\Omega = \{ \Omega^{(n)} \}_{n=0}^\infty \in \Gamma(\mathfrak{h})$ 
is defined by $\Omega^{(0)}=1$ and $\Omega^{(n)}= 0$ ($n \geq 1$), 
which satisfies 
	\begin{equation}
	\label{vacuum}
		c(f)\Omega=0, \quad f \in \mathfrak{h}. 
	\end{equation}
$\Omega$ is a unique vector satisfying \eqref{vacuum}
up to a constant factor.

Let $A$ be a contraction operator on $\mathfrak{h}$, i.e., $\|A\|\leq 1$.
We define a contraction operator $\Gamma(A)$ 
on $\Gamma(\mathfrak{h})$ by
	\[ \left(\Gamma(A)\Psi\right)^{(n)} = \left(\otimes^n A \right)\Psi^{(n)}, 
		\quad \Psi = \{ \Psi^{(n)} \}_{n=0}^{\infty} \]
with the convention $\otimes^0 A=1$.
If $U$ is unitary, i.e. $U^{-1}=U^*$, then $\Gamma(U)$ is also unitary and satisfies
$\Gamma(U)^*=\Gamma(U^*)$ and
	\[ \Gamma(U)c(f)\Gamma(U)^* =c(Uf), \quad \Gamma(U)c^*(f)\Gamma(U)^*=c^*(Uf). \]

For a self-adjoint operator $T$ on $\mathfrak{h}$, i.e., $T=T^*$,
$\{\Gamma(e^{itT})\}_{t \in \mathbb{R}}$ is a strongly continuous one-parameter unitary group
on $\Gamma(\mathfrak{h})$.
Then, by the Stone theorem, there exists a unique self-adjoint operator $d\Gamma(T)$ 
such that
	\[ \Gamma(e^{itT}) = e^{itd\Gamma(T)}. \]
The number operator $N_{\rm f}$ is defined by $d\Gamma(1)$.

\subsection{Bogoliubov transformations}

Let $\mathcal{H}$ be the direct sum of two Hilbert spaces
$\mathcal{H}_+$ and $\mathcal{H}_-$,
where $\mathcal{H}_+ = \mathfrak{h}$ 
and $\mathcal{H}_-$ is a copy of it:
\begin{align*} 
\mathcal{H} 
= \mathcal{H}_+ \oplus \mathcal{H}_- 
= \left\{ v = \begin{bmatrix} v_+ \\ v_- \end{bmatrix} \Bigg| 
v_+, v_- \in \mathfrak{h} \right\}.
\end{align*}
We denote by $P_+$ (resp. $P_-$) 
the projection from $\mathcal{H}$ onto $\mathcal{H}_+$ (resp. $\mathcal{H}_-$):
\[ P_+  \begin{bmatrix} v_+ \\ v_- \end{bmatrix} = \begin{bmatrix} v_+ \\ 0 \end{bmatrix},
\quad P_-  \begin{bmatrix} v_+ \\ v_- \end{bmatrix} = \begin{bmatrix} 0 \\ v_- \end{bmatrix}. \]
A vector $v \in {\rm Ran}P_+ = \mathcal{H}_+$ is identified with one in $\mathfrak{h}$:
$\begin{bmatrix} v_+ \\ 0 \end{bmatrix} = v_+ \in \mathfrak{h}$.
We define an involution $Q$ on $\mathcal{H}$ by
\[ Q = P_+ - P_- \]
and a conjugation $C$ on $\mathcal{H}$ by
\[ C \begin{bmatrix} v_+ \\ v_- \end{bmatrix}
 = \begin{bmatrix} \bar{v}_- \\ \bar{v}_+ \end{bmatrix}, \]
where $\bar{f}$ stands for the complex conjugation of $f$ in $\mathfrak{h}$,
i.e., $\bar{f}(x) = \overline{f(x)}$, a.e. $x \in \mathbb{R}^3$.
 
A bounded operator $A$ on $\mathcal{H}$ is written as
\[ A = \begin{bmatrix} A_{++} & A_{+-} \\ A_{-+} & A_{--} \end{bmatrix}, \]
where $A_{\epsilon \epsilon^\prime} = P_\epsilon A P_{\epsilon^\prime}$ ($\epsilon, \epsilon^\prime = +$ or $-$).
Then we observe that
\[ A \begin{bmatrix} v_+ \\ v_- \end{bmatrix}  
= \begin{bmatrix} A_{++}v_+ + A_{+-}v_- \\ A_{-+}v_+ + A_{--}v_- \end{bmatrix} \]
and $A^*_{\epsilon \epsilon^\prime} = (A_{\epsilon \epsilon^\prime})^*$.
We introduce field operators defined on $\tilde{D}$ by
\begin{align*} 
\psi(v) 
& = c(P_+v) + c^*(CP_-v) \\
& = c(v_+) + c^*(\bar{v}_-), \quad v = \begin{bmatrix} v_+ \\ v_- \end{bmatrix} \in \mathcal{H}.
\end{align*}
One observes that $\psi(v)^* = \psi(Cv)$ on $\tilde{D}$
and hence that $\psi(v)$ is closable.
We denote its closure by the same symbol.
Let $\mathcal{H}_C$ be the set of vectors satisfying $Cv = v$:
\begin{align*} 
\mathcal{H}_C 
& = \left\{  v = \begin{bmatrix} v_+ \\ v_- \end{bmatrix} 
	\Bigg| v_+ = \bar{v}_- \in \mathfrak{h} \right\}.  
\end{align*}
Clearly, for any 
$v = \begin{bmatrix} f \\ \bar{f} \end{bmatrix} \in \mathcal{H}_C$
($f \in \mathfrak{h}$),
the operator $\psi(v)$ is essentially self-adjoint on $D_{\rm f}$
and is equal to $\sqrt{2} \tau(f)$.
By \eqref{936}, we see that 
\[ \tilde{D} = \bigcap_{v \in \mathcal{H}_C} D(\psi(v)). \]
Note that, for any $v = \begin{bmatrix} v_+ \\ v_- \end{bmatrix} \in \mathcal{H}$,
the vectors
\[ v+ Cv = \begin{bmatrix} v_+ + \bar{v}_- \\ v_- + \bar{v}_+ \end{bmatrix},
\quad
i (v-Cv) = \begin{bmatrix} iv_+ -i\bar{v}_- \\ iv_- - i\bar{v}_+ \end{bmatrix} \] 
belong to $\mathcal{H}_C$ and the following holds on $\tilde{D}$:
\begin{align*} 
\psi(v) 
& = \frac{1}{2}\psi(v+Cv) + \frac{i}{2}\psi(i(v-Cv)).
\end{align*}
It is straightforward from \eqref{ccr} that
\[ [\psi(u), \psi(v)^*] = (u,Qv) \]
holds on $D_{\rm f}$.
The following is well known (see, e.g., \cite{Ruijsenaars}):
\begin{lemma} \label{Bogo}
\begin{itemize}
\item[(1)] Let $U$ be a bounded operator on $\mathcal{H}$ satisfying
\begin{equation}
\label{symplectic}
CU = UC, \quad UQU^* = U^*QU = Q 
\end{equation}
Then there exists a unitary operator $\mathscr{U}$ such that  
\[ e^{i\psi(U^*v)} = \mathscr{U}^*e^{i\psi(v)}\mathscr{U},
\quad \mathscr{U}e^{i\psi(v)}\mathscr{U}^* = e^{i\psi(QUQv)}, 
\quad v \in \mathcal{H}_C \]
if and only if $U_{-+}$ is Hilbert-Schmidt.
In this case, $\mathscr{U}$ leaves $\tilde{D}$ invariant.  
\item[(2)]
Let $l$ be a linear functional from $\mathcal{H}$ to $\mathbb{C}$.
Then there exists a unitary operator $\mathscr{U}_l$ such that
\[ e^{i(\psi(u) + l(u))} = \mathscr{U}_l^* e^{i\psi(u)} \mathscr{U}_l,
\quad v \in \mathcal{H}_C \]
if and only if there exists a $u_l \in \mathcal{H}_C$ such that
\[ l(u) = i(v_l,Qu), \quad u \in \mathcal{H}_C. \]
In this case, $\mathscr{U}_l$ leaves $\tilde D$ invariant and $\mathscr{U}_l = e^{-i\psi(v_l)}$.
\end{itemize}
\end{lemma}

\subsection{Quantized Klein Gordon equation}
	
Let $\mathfrak{h} = L^2(\mathbb{R}^3;dx)$ and 
$h_0 = -\Delta$ with the domain $D(h_0) = D(-\Delta) = H^2(\mathbb{R}^3)$.
Then we define the Schr\"odinger operator $h$ by
\begin{align}
h = -\Delta + V(x)
\end{align}
with the potential $V: \mathbb{R}^3 \mapsto \mathbb{R}$
satisfying Assumption \ref{asspot}.
$h$ is self-adjoint with the domain $D(h) = D(h_0)$.
Let 
\[ \omega_0:=(h_0 + m^2)^{1/2} 
\quad \mbox{and} \quad \omega:=(h+m^2)^{1/2}. \]
The free field Hamiltonian $H_{\rm f}$ is defined by
\[ H_{\rm f} = d\Gamma(\omega_0). \]
Since $\omega_0$ and $\omega$ are strictly positive,
$\omega_0^{-1}$ and $\omega^{-1}$ is bounded.
Note that $\omega_0^\theta \omega^{-\theta}$
and
$\omega^{\theta} \omega_0^{-\theta}$ 
are bounded operators for any $ 0 \leq \theta \leq 1$. 
Indeed, since $h_0 \leq h$, we observe that 
$\|\omega_0^{\theta} f\| \leq \|\omega^{\theta} f\|$.
Hence $\omega_0^\theta \omega^{-\theta}$ is bounded.
On the other hand, since $D(h) = D(h_0)$, 
it follows from the closed graph theorem that
$\|h f\| \leq C \|(h_0 + 1) f\|$ with some $C>0$.
Hence we observe that $\omega^{\theta}\omega_0^{-\theta}$ is bounded. 
From this fact, we see that $\omega^{-\theta} \omega_0^{\theta}$ and $\omega_0^{-\theta} \omega^{\theta}$ 
can be extended to bounded operators on $\mathfrak{h}$.
We denote the extended operators 
by the same symbols.
The following holds.
\begin{lemma}
\label{HSlemma}
Suppose that Assumption \ref{asspot} holds. Then
the operators $\omega_0^{1/2}\omega^{-1/2} - 1$,
$\omega_0^{-1/2}\omega^{1/2} - 1$ 
and $\omega_0^{-1/2}(\omega_0 - \omega ) \omega_0^{-1/2}$ are Hilbert-Schmidt.
\end{lemma} 
\begin{proof}
See Appendix A.1.
\end{proof}

For real $f \in H^{-1/2}(\mathbb{R}^3)$
and $g \in H^{1/2}(\mathbb{R}^3)$,
we set
\begin{align*}
\phi_0(f) = \psi \left( u_0 \right)
\quad \mbox{and} \quad  
\pi_0(g) = \psi \left( v_0 \right),
\end{align*}
where 
\begin{align*}
u_0 = \begin{bmatrix} \omega_0^{-1/2} f/\sqrt{2} \\ \omega_0^{-1/2} f/\sqrt{2}  \end{bmatrix} 
\quad \mbox{and} \quad 
v_0 = \begin{bmatrix} i\omega_0^{+1/2} g/\sqrt{2} \\ -i\omega_0^{+1/2} g/\sqrt{2}  \end{bmatrix}.
\end{align*}
For non real $f \in H^{-1/2}(\mathbb{R}^3)$ 
(resp. $g \in H^{1/2}(\mathbb{R}^3)$),
$\phi_0(f)$ (resp. $\pi_0(g)$) is defined by
of $\phi_0(({\rm Re}f) + i\phi_0({\rm Im}f)$ 
(resp. $\pi_0(({\rm Re}g) + i\pi_0({\rm Im}g)$ ).
Note that for non real 
$f \in H^{-1/2}(\mathbb{R}^3)$ and $g \in H^{1/2}(\mathbb{R}^3)$,
$\phi_0(f)$ and $\pi_0(g)$ are non self-adjoint and
the following equations hold on $\tilde D$:
\begin{align*} 
& \phi_0(f)
=
\frac{1}{\sqrt{2}} 
	\left(c^*(\omega_0^{-1/2}f)
		+ c(\omega_0^{-1/2} \bar f) \right), \\
& \pi_0(g)
=  
\frac{i}{\sqrt{2}} 
	\left(c^*(\omega_0^{+1/2}g) 
		- c(\omega_0^{+1/2} \bar g) \right).
\end{align*}
By \eqref{ccr}, one can show that
the CCR holds on $D_{\rm f}$:
\begin{align}
\label{fccr}
[ \phi_0(f), \pi_0(g)] = i(\bar{f},g),
\quad [\phi_0(f),\phi_0(\tilde f)] = [\pi_0(g),\pi_0(\tilde g)] = 0,
\end{align} 
for any $f, \tilde f \in H^{-1/2}(\mathbb{R}^3)$
and $g, \tilde g \in H^{1/2}(\mathbb{R}^3)$.
It holds from \eqref{fccr} that
\begin{align*} 
& \|\phi_0(f)\Psi\|^2 
= \|\phi_0({\rm Re}f) \Psi\|^2 + \|\phi_0({\rm Re}f)\|^2, \\
& \|\pi_0(g)\Psi\|^2 
= \|\pi_0({\rm Re}g) \Psi\|^2 + \|\pi_0({\rm Re}g)\|^2, 
\quad \Psi \in D_{\rm f}.  
\end{align*}
Hence we observe that, for non real 
$f \in H^{-1/2}(\mathbb{R}^3)$
and $g \in H^{1/2}(\mathbb{R}^3)$,
$\phi_0(f)$ and $\pi_0(g)$ are closed on the natural domain.

We introduce the bounded operator on $\mathcal{H}$ by
\begin{align*} 
U(t) = \begin{bmatrix} U_{++}(t) & U_{+-}(t) \\ U_{-+}(t) & U_{--}(t) \end{bmatrix},
\quad t \in \mathbb{R},
\end{align*}  
where 
\begin{align*}
U_{++}(t) 
& = \frac{1}{2} \Big( \omega_0^{-1/2} \cos(t\omega) \omega_0^{1/2} 
		+ \omega_0^{1/2} \cos(t\omega) \omega_0^{-1/2} \Big) \\
& \quad - \frac{i}{2} \Big( \omega_0^{1/2} \omega^{-1/2} \sin(t\omega) \omega^{-1/2}\omega_0^{1/2} \\
& \quad + \omega_0^{-1/2} \omega^{1/2} \sin(t\omega) \omega^{1/2}\omega_0^{-1/2} \Big)
\end{align*}
and 
\begin{align*}
U_{-+}(t) 
& = \frac{1}{2} \Big( \omega_0^{1/2} \cos(t\omega) \omega_0^{-1/2} 
	- \omega_0^{-1/2} \cos(t\omega) \omega_0^{1/2} \Big) \\
& \quad - \frac{i}{2} \Big( \omega_0^{1/2} \omega^{-1/2} \sin(t\omega) \omega^{-1/2}\omega_0^{1/2} \\
& \quad - \omega_0^{-1/2} \omega^{1/2} \sin(t\omega) \omega^{1/2}\omega_0^{-1/2}\Big)
\end{align*}
with $U_{--}(t) = \overline{U_{++}(t)}$ and $U_{+-}(t) = \overline{U_{-+}(t)}$.
Here, for a linear operator $A$, we define $\bar{A}$ by $ \bar{A} f = \overline{A \bar{f}}$.  

\begin{lemma}
Suppose that Assumption \ref{asspot} holds.
Then there exists a family of unitary operators 
$\mathscr{U}_t$ on $\Gamma(\mathfrak{h})$ 
such that 
$\mathscr{U}_t$ maps $\tilde{D}$ to $\tilde{D}$ and for $v \in \mathcal{H}_C$
\begin{align*} 
e^{i\psi(U(t)^*v)} = \mathscr{U}_t^*e^{i\psi(v)}\mathscr{U}_t,
\quad \mathscr{U}_te^{i\psi(v)}\mathscr{U}_t^* = e^{i\psi(QU(t)Qv)}. 
\end{align*}
\end{lemma}
\begin{proof}
By direct calculation, we observe that 
$U(t)$ satisfies \eqref{symplectic}
with $U= U(t)$ 
.
We note that $2U_{-+}(t)$ is equal to 
\begin{align*}
& (\omega_0^{1/2}\omega^{-1/2} - 1) 
	\cos(t\omega) \omega^{1/2} \omega_0^{-1/2} \\
& \quad + \cos(t\omega) \cdot \omega^{-1/2} \omega_0^{1/2} \cdot 
	\omega_0^{-1/2}(\omega-\omega_0)\omega_0^{-1/2} \\
& \quad 
	+ (1-\omega_0^{-1/2}\omega^{1/2}) \cos(t\omega) \omega^{-1/2}\omega_0^{1/2} \\
& \quad
	-i (\omega_0^{1/2} \omega^{-1/2}-1) \sin(t\omega) 
		\cdot \omega^{-1/2}\omega_0^{1/2} \\
& \quad -i \sin(t\omega) \cdot \omega^{-1/2} \omega_0^{1/2}
		\cdot \omega_0^{-1/2}(\omega_0 - \omega ) \omega_0^{-1/2} \\
& \quad -i (1-\omega_0^{-1/2} \omega^{1/2}) 
		\sin(t\omega) \cdot \omega^{1/2}\omega_0^{-1/2}.
\end{align*}
By Lemma \ref{HSlemma},
$\omega_0^{1/2}\omega^{-1/2} - 1$,
$\omega_0^{-1/2}\omega^{1/2} - 1$ 
and $\omega_0^{-1/2}(\omega_0 - \omega ) \omega_0^{-1/2}$ are Hilbert-Schmidt,
so is $U_{-+}(t)$.
\end{proof}

Let
\begin{align*} 
g_t & = -\frac{1}{\sqrt{2}} \int_0^t ds \omega_0^{-1/2} \cos[(t-s)\omega] J_s \\
	& \quad + \frac{i}{\sqrt{2}}\int_0^t ds \omega_0^{1/2} \frac{\sin[(t-s)\omega]}{\omega} J_s
\end{align*}
and
\[ j_t = \begin{bmatrix} g_t \\ \bar{g_t} \end{bmatrix}.
\]
For $f \in H^{-1/2}(\mathbb{R}^3)$
and $g \in H^{1/2}(\mathbb{R}^3)$,
we define field operators 
\begin{align*}
\phi(t,f) = \mathscr{U}(t)^* \phi_0(f)\mathscr{U}(t)
\quad \mbox{and} \quad
\pi(t,g) = \mathscr{U}(t)^* \pi_0(g)\mathscr{U}(t),
\end{align*}
where 
\[ \mathscr{U}(t) = e^{-i\psi(j_t)} \mathscr{U}_t. \]
The following propositions are standard: 
\begin{proposition}
\label{repKG}
Suppose that Assumptions \ref{asspot} 
and \ref{assext} hold.
Then
it holds that,
$f \in H^{-1/2}(\mathbb{R}^3)$
and $g \in H^{1/2}(\mathbb{R}^3)$,
\begin{align*}
& \phi(t,f) = \psi\left( u(t) \right) 
	+ \left(\int_0^t ds \frac{\sin[(t-s)\omega]}{\omega}J_s,f \right), \\ 
& \pi(t,f) = \psi\left( v(t) \right)
	- \left(\int_0^t ds \cos[(t-s)\omega]J_s,f \right), 
\end{align*}
where
\begin{align*} 
& u(t) 
= \begin{bmatrix}
\omega^{-1/2}_0 (\cos(t\omega)
	+i\omega_0 \sin(t\omega)\omega^{-1})\bar{f} /\sqrt{2} \\
\omega^{-1/2}_0 (\cos(t\omega)
	-i\omega_0 \sin(t\omega)\omega^{-1}) \bar{f} /\sqrt{2} 
\end{bmatrix}, \\
& v(t) 
= \begin{bmatrix} 
i \omega^{1/2}_0 (\cos(t\omega) 
	+ i\omega_0^{-1}\sin(t\omega)\omega) \bar{f}/\sqrt{2} \\
-i \omega^{1/2}_0 (\cos(t\omega) 
	- i\omega_0^{-1}\sin(t\omega)\omega) \bar{f}/\sqrt{2} 
\end{bmatrix}.
\end{align*}
\end{proposition}

\begin{proposition}
\label{opdis}
Suppose that Assumptions \ref{asspot} 
and \ref{assext} hold.
Let $\Psi \in D(N_{\rm f}^{1/2})$
and $f \in H^{-1/2}(\mathbb{R}^3) \cap H^{3/2}(\mathbb{R}^3)$. 
Then
$\phi(0,f) = \phi_0(f)$, $\pi(0,f) = \pi_0(f)$ and
\begin{align*} 
& \frac{d}{dt}\phi(t,f)\Psi = \pi(t,f)\Psi, \\
& \frac{d^2}{dt^2}\phi(t,f) \Psi+ \phi(t, (m^2 -\Delta)f) \Psi = (J_t, f).  
\end{align*}
\end{proposition}

\begin{remark}
\label{rem2.1}
Let $\Psi \in D(N_{\rm f}^{1/2})$.
Then
$\mathcal{S}(\mathbb{R}^d) \ni f \mapsto (\Psi,\phi(t,f) \Psi)$
and $\mathscr{S}(\mathbb{R}^d) \ni g \mapsto (\Psi,\pi(t,g) \Psi)$ are 
tempered distributions and
symbols $\varphi_\Psi(t,x)$ 
and $\varpi_\Psi(t,x)$,
defined formally as  
\[ \int \varphi_\Psi(t,x) f(x) dx = (\Psi,\pi(t,f)\Psi)
\quad \mbox{and} \quad
\int \varpi_\Psi(t,x) f(x) dx = (\Psi, \pi(t,f)\Psi),\] 
satisfy
\[ \varphi_\Psi(t,x) \in H^{1/2}(\mathbb{R}^3), \quad
	\varpi_\Psi(t,x) \in H^{-1/2}(\mathbb{R}^3). \]
We denote by $\Gamma_0$ the linear hull of 
\[ \{\Omega\} \cup \{ c^*(f_1) \cdots c^*(f_n)\Omega \mid 
	f_j \in D(\omega_0), ~j=1, \cdots, n, ~n \geq1 \}. \]
Note that $\Gamma_0$ is dense in $\Gamma(\mathfrak{h})$
and that
$\phi(t,f)$ and $\pi(t,f)$ ($f \in \mathcal{S}(\mathbb{R}^3)$)
leave $\Gamma_0$ invariant.
Suppose that $J_t \in H^{1/2}(\mathbb{R}^3)$ holds.
If $\Psi$ is a vector belonging to $\Gamma_0$,
then $\varphi_\Psi(t) \in H^{3/2}(\mathbb{R}^3)$,
$\varpi_\Psi(t) \in H^{1/2}(\mathbb{R}^3)$ and
\begin{align*}
\begin{bmatrix} \varphi_\Psi(t) \\ \varpi_\Psi(t) \end{bmatrix}
& = \begin{bmatrix} \cos (t\omega) & \omega^{-1}\sin (t\omega) \\
		-\omega \sin(t\omega) & \cos(t\omega) \end{bmatrix}
\begin{bmatrix} \varphi_\Psi(0) \\ 
	\varpi_\Psi(0)  \end{bmatrix} \\
& \quad 
+ \begin{bmatrix} \int_0^t ds \sin[(t-s)\omega] \omega^{-1}J_s \\ 
	 -\int_0^t ds \cos[(t-s)\omega] J_s\end{bmatrix},
\end{align*}
which gives a solution of 
\begin{align}
\label{CKGE}
i \frac{d}{dt} \begin{bmatrix} \varphi(t) \\ \varpi(t) \end{bmatrix}
& = \begin{bmatrix} 0 & i \\ -i\omega^2 & 0 \end{bmatrix}
		\begin{bmatrix} \varphi(t) \\ \varpi(t) \end{bmatrix}
		+ \begin{bmatrix} 0 \\ J_t \end{bmatrix}
\end{align}
with the initial value 
\begin{align}
\label{IVCKGE} 
\begin{bmatrix} \varphi(0) \\ \varpi(0) \end{bmatrix}
= \begin{bmatrix} \varphi_\Psi(0) \\ 
	\varpi_\Psi(0)  \end{bmatrix}.
\end{align}
\end{remark}
We note that $\Psi \in \Gamma_\infty$ 
belongs to $D(N^{1/2})$
and is an analytic vector of $\phi(t,f)$ and $\pi(t,f)$,
i.e., for any $0< t \leq t_0$,
\[ \sum_{n=0}^\infty \frac{t^n}{n!} \|\phi(t,f)^n \Psi\| < \infty,
\quad \mbox{and} \quad 
\sum_{n=0}^\infty \frac{t^n}{n!}\|\pi(t,f)^n \Psi\| < \infty \]
with some $t_0>0$.

\begin{remark}
\label{rem2.2}
Suppose that $J_t \in H^{1/2}(\mathbb{R}^3)$ holds.
Then the pair of field operators $\phi(t,f)$ and $\pi(t,f)$ is unique in the following sense:
If there exist a pair of field operators
$\phi^\prime(t,f) = \mathscr{U}^\prime(t)^* \phi_0(f)\mathscr{U}^\prime(t)$ and $\pi^\prime(t,f) = \mathscr{U}^\prime(t)^* \pi_0(f)\mathscr{U}^\prime(t)$
with a family of unitary operators $\mathscr{U}^\prime(t)$
satisfying the following conditions (1) - (4),
then $\phi(t,f) = \phi^\prime(t,f)$ and $\pi(t,f) = \pi^\prime(t,f)$:
\begin{itemize}
\item[(1)] $\mathscr{U}^\prime(0)=I$.
\item[(2)] $\phi^\prime(t,f)$ and $\pi^\prime(t,f)$ leave $\Gamma_0$ invariant.
\item[(3)] The vector $\Psi \in \Gamma_0$
is an analytic vector of $\phi^\prime(t,f)$ and $\pi^\prime(t,f)$.
\item[(4)] The distributional kernels 
$\varphi_\Psi^\prime(t)$ and $\varpi_\Psi^\prime(t)$ of 
\begin{align*} 
& \varphi_\Psi^\prime(t,f) 
= (\Psi, \phi^\prime(t,f) \Psi), \\
& \varpi_\Psi^\prime(t,f) 
= (\Psi, \pi^\prime(t,f) \Psi)
\end{align*}
with $\Psi \in \Gamma_0$ 
satisfy the equation \eqref{CKGE}
with the initial value \eqref{IVCKGE}.
\end{itemize}
The uniqueness can be proved as follows.
Using the conditions (1), (4) and the uniqueness of 
the solution of \eqref{CKGE},
we infer
\[ \phi(t,f) = \phi^\prime(t,f), \quad \pi(t,g) = \pi^\prime(t,g) \]
on $\Gamma_0$.
By the conditions (2) and (3), 
we have
\begin{align*} 
e^{it\phi(t,f)}\Psi 
= \sum_{n=0}^\infty \frac{(it)^n}{n_!} \phi(t,f)^n \Psi
= \sum_{n=0}^\infty \frac{(it)^n}{n_!} \phi^\prime(t,f)^n \Psi 
= e^{i\phi^\prime(t,f)}\Psi
\end{align*}
for $\Psi \in \Gamma_0$ and sufficiently small $t>0$.
Since $\Gamma_0$ is dense,
we observe that  $e^{i\psi(t,f)} = e^{i\phi^\prime(t,f)}$
for real $f \in H^{-1/2}(\mathbb{R}^3)$.
By the uniqueness of the generator,
we have the operator equality $\psi(t,f) = \phi^\prime(t,f)$ 
for  real $f \in H^{-1/2}(\mathbb{R}^3)$.
Note that $\phi(t,f)$ and $\phi^\prime(t,f)$ 
are unitary equivalent to $\phi_0(f)$.
By the similar argument as in the proof of the closedness of $\phi_0(f)$,
one can prove that the following operator equation holds
for non real $f \in H^{-1/2}(\mathbb{R}^3)$:
\begin{align*} 
\phi(t,f) 
& = \phi(t,{\rm Re}f) + i\phi(t,{\rm Im}f) \\
& = \phi^\prime(t,{\rm Re}f) + i\phi^\prime(t,{\rm Im}f) 
= \phi^\prime(t,f)
\end{align*}
$\pi(t,f) = \pi^\prime(t,f)$ is proved similarly.
\end{remark}

\subsection{Wave operators}

Let
\[ \mathscr{U}_0(t) = e^{-itH_{\rm f}} \]
and
\[ U_0(t) = \begin{bmatrix} e^{-it\omega_0} & 0 \\ 0 & e^{it\omega_0} \end{bmatrix}. \]
One observe that
\[ \mathscr{W}(t) = \mathscr{U}(t)^* \mathscr{U}_0(t) \]
satisfies
\begin{align}
\mathscr{W}(t) e^{i\psi(v)}\mathscr{W}(t)^*
& = \mathscr{U}_t^* e^{i(\psi(U_0(t)v)+i(j_t,QU_0(t)v))} \mathscr{U}_t \\
& \label{wt}
= e^{i\psi((U(t)^*U_0(t)v)+i(U_0(t)^*j_t,Qv))} 
\end{align}
for $v \in \mathcal{H}_C$.
By the Stone theorem, we have
\begin{align}
\mathscr{W}(t) \psi(v)\mathscr{W}(t)^*
= \psi(U(t)^*U_0(t)v) + i(U_0(t)^* j_t, Qv).
\end{align}

\begin{lemma}
\label{waveop}
Suppose that Assumptions \ref{asspot} and \ref{assscat} hold.
Then:
\[ \mbox{{\rm s}-} \lim_{t \to \pm \infty}U(t)^*U_0(t) = W_{\pm}, \]
where
\[ W_{\pm} 
= 
\begin{bmatrix}
(W_\pm)_{++} & (W_\pm)_{+-} \\
(W_\pm)_{-+} & (W_\pm)_{--} 
\end{bmatrix} \]
with
\begin{align*} 
(W_\pm)_{++}
& = \overline{(W_\pm)_{--}} \\
& = \frac{1}{2}(\omega_0^{-1/2} w_\pm \omega_0^{+1/2}
	+ \omega_0^{+1/2} w_\pm \omega_0^{-1/2}), \\
(W_\pm)_{-+}
& = \overline{(W_\pm)_{+-}} \\
& = \frac{1}{2}(\omega_0^{-1/2} w_\pm \omega_0^{+1/2}
	- \omega_0^{+1/2} w_\pm \omega_0^{-1/2}).
\end{align*} 
\end{lemma}
\begin{proof}
See Appendix A.3.
\end{proof}

\begin{lemma}
\label{dua}
Suppose that Assumptions \ref{asspot}, \ref{assscat} and \ref{assext} hold.
Then:
\begin{align*} 
\lim_{t \to \pm \infty} U_0(t)^*j_t = j_\pm,
\end{align*}
where 
\begin{align*}
j_\pm = \begin{bmatrix} g_\pm \\ \bar g_\pm \end{bmatrix}
\quad 
\mbox{and}
\quad
g_\pm = -\frac{1}{\sqrt{2}} \int_0^{\pm \infty} ds 
	\omega_0^{-1/2} e^{is\omega_0} w_\pm^* J_s.
\end{align*}
\end{lemma}
\begin{proof}
It suffices to prove that
\begin{align}
\label{1405}
\lim_{t \to \pm \infty}e^{it\omega_0}g_t
= -\frac{1}{\sqrt{2}} \int_0^{\pm \infty} ds 
	\omega_0^{-1/2} e^{is\omega_0} w_\pm^* J_s.
\end{align}
By a direct calculation, we have
\begin{align}
e^{it\omega_0}g_t
& = -\frac{1}{2\sqrt{2}} \int_0^t ds e^{it\omega_0}
	\omega_0^{-1/2} \left(\omega - \omega_0\right) \omega_0^{-1/2}
		\cdot \omega_0^{1/2} \omega^{-1/2} \notag \\
& \quad \quad \quad \quad \quad  \times 
	e^{i(t-s)\omega} \omega^{-1/2}J_s \notag \\
\label{1532}
& \quad 
-\frac{1}{2\sqrt{2}} \int_0^t ds e^{it\omega_0}
	\left( \omega_0^{-1/2} + \omega_0^{1/2}\omega^{-1}\right) e^{-i(t-s)\omega}J_s.
\end{align}
Since 
$\omega_0^{-1/2} \left(\omega - \omega_0\right) \omega_0^{-1/2}$
is Hilbert-Schmidt by Lemma \ref{HSlemma}
and since $\int_0^t ds \|\omega^{-1/2}J_s\| < \infty$
by (b) of Assumption \ref{assext}, 
the first term of the r.h.s in \eqref{1532} tends to zero
as $t$ goes to $\pm \infty$.
We show that
the limit of the second term in \eqref{1532} equals \eqref{1405}.
It holds that
\begin{align} 
& e^{it\omega_0} 
\left( \omega_0^{-1/2} 
	+ \omega_0^{1/2}\omega^{-1}\right) e^{-it\omega} \notag \\
& = \left( \omega_0^{-1/2} e^{it\omega_0}e^{-it\omega} 
	+ e^{it\omega_0}e^{-it\omega} \omega^{-1/2} \right) \notag \\
\label{1544}
& \quad 
	+ e^{it\omega_0}(\omega_0^{-1/2} - \omega^{-1/2})\omega^{-1}e^{-it\omega}. 
\end{align}
The first term of the r.h.s in \eqref{1544} 
tends to $2 \omega_0^{-1/2} w_\pm^*$.
The second term tends to zero
since
\begin{align*}
(\omega_0^{-1/2} - \omega^{-1/2})\omega^{-1}e^{-it\omega}
= \omega_0^{-1/2}(1-\omega_0^{1/2}\omega^{-1/2})\omega^{-1}e^{-it\omega}
\end{align*}
and since, by Lemma \ref{HSlemma}, 
$(1-\omega_0^{1/2}\omega^{-1/2})$ is Hilbert-Schmidt.
Using these facts and Assumption \ref{assext} (b) again,
we infer
the limit of the second term in \eqref{1532} equals \eqref{1405}.
\end{proof}

By Lemmas \ref{waveop} and \ref{dua}, we have
\begin{lemma}
\label{1031}
Suppose that Assumptions \ref{asspot}, \ref{assscat} 
and \ref{assext} hold.
Then it holds that, for $v \in \mathcal{H}_C$,
\begin{equation}
\label{1712}
\mbox{s-}\lim_{t \to \pm}\mathscr{W}(t) e^{i \psi(v)} \mathscr{W}(t)^* 
	= e^{i ( \psi(W_\pm v) +i(j_\pm, Qv) )}.
\end{equation}
In particular, the following properties hold:
for real $f \in H^{-1/2}(\mathbb{R}^3)$
and $g \in H^{1/2}(\mathbb{R}^3)$,
\begin{align*}
& \mbox{s-}\lim_{t \to \pm} \mathscr{W}(t) e^{i \phi_0(f)} \mathscr{W}(t)^* 
	= e^{i \phi_\pm(f)}, \\
& \mbox{s-}\lim_{t \to \pm} \mathscr{W}(t) e^{i \pi_0(g)} \mathscr{W}(t)^*
	= e^{i \pi_\pm(g)}, 
\end{align*}
where
\begin{align*} 
& \phi_\pm(f) 
	= 
\psi(W_\pm u_0)
+ i(j_\pm, Qu_0), \\
& \pi_\pm(f)
	= 
\psi(W_\pm v_0)
+ i(j_\pm, Qv_0).
\end{align*}
\end{lemma}
\begin{proof}
Since $e^{i\psi(v)}$ ($v \in \mathcal{H}_C$) is unitary,
it suffices to prove \eqref{1712} on $D_{\rm f}$,
which is an easy exercise.
\end{proof}

\begin{lemma} 
\label{waveop2}
Suppose that Assumptions \ref{asspot}, \ref{assscat} and \ref{assext}
hold.
Then there exists a unitary operator $\mathscr{W}_\pm$ such that
\[ 
\mathscr{W}_\pm e^{i\psi(v)} \mathscr{W}_\pm^* 
= e^{i ( \psi(W_\pm v) +i(j_\pm, Qv) )},
\quad v \in \mathcal{H}_C.
\]
In particular, it holds that:
for $f \in H^{-1/2}(\mathbb{R}^3)$
and $g \in H^{1/2}(\mathbb{R}^3)$,
\begin{align*}
& \mathscr{W}_\pm \phi_0(f) \mathscr{W}_\pm^* = \phi_\pm(f), \\
& \mathscr{W}_\pm \pi_0(f) \mathscr{W}_\pm^* = \pi_\pm(f). 
\end{align*}
\end{lemma}

\begin{proof}
Note that
$(W_\pm^*)_{-+} = W_{+-}^*$ is equal to
\begin{align*} 
& \frac{1}{2}\omega_0^{1/2} w_\mp^* \omega_0^{-1/2} - \omega_0^{-1/2} w_\mp^* \omega_0^{1/2}) \\
& \quad =  \frac{1}{2} \left(w_\mp^* \omega^{-1/2}\omega_0^{1/2}\right)
	\cdot \left(\omega_0^{-1/2} (\omega - \omega_0) \omega_0^{-1/2}\right)
\end{align*}
and hence is Hilbert-Schmidt.
By Lemma \ref{Bogo}, it holds that there exists a unitary operator $\mathscr{U}_\pm$ such that
\begin{align*} 
\mathscr{U}_\pm e^{i\psi(v)} \mathscr{U}_\pm^* = e^{i\psi(W_\pm v)}, 
\quad v \in \mathcal{H}_C.
\end{align*}
Setting $\mathscr{W}_\pm = \mathscr{U}_\pm e^{i\psi(j_\pm)}$, one obtains the desired result.
\end{proof}


\section{Scattering theory}
Throughout this section, we suppose that Assumptions \ref{asspot}
- \ref{assext} hold.
Since any two quantum systems which are transformed from one to the other 
by a unitary transformation are physically equivalent, 
one can redefine the solution of the Klein-Gordon field
as
\[ {\bm \phi}(t,f) = \mathscr{W}_-^* \phi(t,f) \mathscr{W}_-,
\quad {\bm \pi}(t,g) = \mathscr{W}_-^* \pi(t,g) \mathscr{W}_- \]
for $f \in H^{-1/2}(\mathbb{R}^3)$
and $g \in H^{1/2}(\mathbb{R}^3)$.
Then it follows from Proposition \ref{repKG} 
(see also Remarks \ref{rem2.1} - \ref{rem2.2}) that
${\bm \phi}(t,f)$ and ${\bm \pi}(t,g)$ satisfy \eqref{KG}
in the operator valued distributional sense.
The field operators ${\bm \phi}_s(t,f)$
and ${\bm \pi}_s(t,g)$ defined by
\begin{align*} 
& {\bm \phi}_s(t,f) 
= \mathscr{W}_-^* \mathscr{U}(s)^* \mathscr{U}(s-t) 
	\phi_0(f)\mathscr{U}(s-t) ^*\mathscr{U}(s)\mathscr{W}_-, \\
& {\bm \pi}_s(t,g) 
= \mathscr{W}_-^* \mathscr{U}(s)^* \mathscr{U}(s-t) 
	\pi_0(g)\mathscr{U}(s-t) ^*\mathscr{U}(s)\mathscr{W}_-
\end{align*}
satisfy the free Klein-Gordon equation with the initial condition:
\[ {\bm \phi}_s(s,f) = {\bm \phi}(s,f),
\quad {\bm \pi}_s(t,g) = {\bm \pi}(s,g). \]
The asymptotic fields ${\bm \phi}_{\sharp}(t,f)$ and ${\bm \pi}_{\sharp}(t,g)$ 
($\sharp$ stands for ${\rm out}$ or ${\rm in}$) are defined as
\begin{align*} 
e^{i{\bm \phi}_{\rm out/in}(t,f)} = \lim_{t \to \pm \infty}e^{i{\bm \phi}_s(t,f)},
\quad e^{i{\bm \pi}_{\rm out/in}(t,g)} = \lim_{t \to \pm \infty}e^{i{\bm \pi}_s(t,g)} 
\end{align*}
for any real $f \in H^{-1/2}(\mathbb{R}^3)$ and $g \in H^{1/2}(\mathbb{R}^3)$.
Then, by Lemmas \ref{1031} and \ref{waveop2}, the incoming fields 
${\bm \phi}_{\rm in}(f) = {\bm \phi}_{\rm in}(0,f) $ and ${\bm \pi}_{\rm in}(g) = {\bm \pi}_{\rm in}(0, g) $ are 
\begin{equation}
\label{1538in} 
{\bm \phi}_{\rm in}(f) = \phi_0(f), \quad {\bm \pi}_{\rm in}(g)=\pi_0(g) 
\end{equation}
and
\[ \{ {\bm \phi}_{\rm in}(f), {\bm \pi}_{\rm in}(f) 
\mid f \in H^{-1/2}(\mathbb{R}^3), ~g \in H^{1/2}(\mathbb{R}^3) \}
\] 
gives the Fock representation of the CCR (see, e.g., \cite{Ar}).
The out going fields ${\bm \phi}_{\rm out}(f) = {\bm \phi}_{\rm out}(0, f)$ 
and ${\bm \pi}_{\rm out}(g) = {\bm \pi}_{\rm out}(0, g)$ are 
\begin{equation}
\label{1538out} 
{\bm \phi}_{\rm out}(f) = (\mathscr{W}_+^* \mathscr{W}_-)^* \phi_0(f)(\mathscr{W}_+^* \mathscr{W}_-), 
\quad {\bm \pi}_{\rm in}(g)= (\mathscr{W}_+^* \mathscr{W}_-)^*\pi_0(g)(\mathscr{W}_+^* \mathscr{W}_-). 
\end{equation}

The scattering matrix $\mathscr{S} = \mathscr{S}(V,J)$ is defined by
\[ \mathscr{S} = \mathscr{W}_+^* \mathscr{W}_-. \]
\begin{proposition}
Suppose that Assumptions \ref{asspot}, \ref{assscat} 
and \ref{assext} hold.
Then:
\[ \mathscr{S}^{-1} {\bm \phi}_{\rm in}(f) \mathscr{S} 
= {\bm \phi}_{\rm out}(f)
\quad \mbox{and} \quad
\mathscr{S}^{-1} {\bm \pi}_{\rm in}(f) \mathscr{S} 
= {\bm \pi}_{\rm out}(f). \]
\end{proposition}
\begin{proof}
By \eqref{1538in} and \eqref{1538out}, we see that
\[ 
\mathscr{S}^{-1} \psi_{\rm in}(f) \mathscr{S} 
= \mathscr{W}_-^* \mathscr{W}_+ \psi_0(f) \mathscr{W}_+^* \mathscr{W}_-
= \psi_{\rm out}(f), 
\]
where $\psi_{{\rm in/out},0}$ stands for ${\bm \phi}_{{\rm in/out},0}$ 
or ${\bm \pi}_{{\rm in/out},0}$.
\end{proof}

Let us define the associated annihilation and creation operators by 
\[ c_{\rm in}(f) = c(f) \quad 
\mbox{and} \quad c_{\rm out}(f) = \mathscr{S}^* c(f) \mathscr{S}
 \]
and $c_\sharp^*(f) = c_\sharp(f)^*$.
The free Hamiltonian of the incoming field and
outgoing field are defined by
\[
H_{\rm in} = d\Gamma(\omega_0)
\quad \mbox{and} \quad
H_{\rm out} 
= \mathscr{S}^* 
	d\Gamma(\omega_0) \mathscr{S}.
\]
The asymptotic vacua 
\[ \Omega_{\rm in} = \Omega 
\quad \mbox{and} \quad
	\Omega_{\rm out} = \mathscr{S}^*\Omega \]
satisfy
\[ H_\sharp \Omega_\sharp = 0 
\quad \mbox{and} \quad c_\sharp(f) \Omega_\sharp = 0. \]
The asymptotic fields satisfy
\begin{align*} 
{\bm \phi}_\sharp(t,f) 
= e^{it H_\sharp} {\bm \phi}_\sharp(f) e^{-it H_\sharp}, 
\quad \mbox{and} \quad
{\bm \pi}_\sharp(t,f)
= e^{it H_\sharp} {\bm \pi}_\sharp(f) e^{-it H_\sharp}, 
\end{align*}
and 
\begin{align*} 
& {\bm \phi}_\sharp(t,f)
=
\frac{1}{\sqrt{2}} 
	\left(c_\sharp^*(e^{it\omega_0}\omega_0^{-1/2}f)
		+ c_\sharp(e^{it\omega_0}\omega_0^{-1/2} \bar f) \right), \\
& {\bm \pi}_\sharp(t,f)
=  
\frac{i}{\sqrt{2}} 
	\left(c_\sharp^*(e^{it\omega_0}\omega_0^{+1/2}f) 
		- c_\sharp(e^{it\omega_0}\omega_0^{+1/2} \bar f) \right)
\end{align*}
hold on $\tilde D$.


\section{Inverse scattering}

\subsection{Uniqueness of the potential $V$}
Let 
\[ S = w_+^* w_-. \]
When we want to emphasize the dependence of $V$,
we write $S = S(V)$. 
We will prove the following theorem:
\begin{theorem}
\label{main1}
Suppose that Assumptions \ref{asspot}, \ref{assscat} and \ref{assext} hold.
If $\mathscr{S}(V,J) = \mathscr{S}(V^\prime,J^\prime)$,
then $S(V)=S(V^\prime)$.
\end{theorem}
We need the following:
\begin{lemma}
\label{tec1}
Suppose that Assumptions \ref{asspot}, \ref{assscat} 
and \ref{assext} hold.
Then:
\begin{equation}
\label{0845} 
\lim_{t \to \pm \infty} 
	(c^*(e^{it\omega_0}f) \Omega,	
		\mathscr{S} c^*(e^{it \omega_0}g) \Omega) 
	= (\Omega, \mathscr{S} \Omega) (f, S g). 
\end{equation}
\end{lemma}

\begin{proof}[Proof of Theorem \ref{main1}]
For notational simplicity, we write
$\mathscr{S}(V^\prime,J^\prime) = \mathscr{S}^\prime$
and $S(V^\prime) = S^\prime$.
Note that $(\Omega, \mathscr{S}\Omega) \not=0$
since $\mathscr{S}$ is unitary.
By Lemma \ref{tec1}, we have
\begin{align*} 
(f,Sg) 
& = \lim_{t \to \pm \infty}(c^*(e^{it\omega_0}f) \Omega,	
		\mathscr{S} c^*(e^{it \omega_0}g) \Omega)/
		(\Omega, \mathscr{S}\Omega) \\
& = \lim_{t \to \pm \infty}(c^*(e^{it\omega_0}f) \Omega,	
		\mathscr{S}^\prime c^*(e^{it \omega_0}g) \Omega)/
		(\Omega, \mathscr{S}^\prime\Omega)
& = (f,S^\prime g).
\end{align*}
\end{proof}
In the remainder of this subsection,
we will prove Lemma \ref{tec1}.
Let
\[ \mathscr{I}_t(f,g)
= (c^*(e^{it\omega_0}f) \Omega,	
		\mathscr{S} c^*(e^{it \omega_0}g) \Omega).  
\]
We see from Lemma \ref{waveop2} that 
$\mathscr{W}_\pm = \mathscr{U}_\pm e^{i\psi(j_\pm)}$.
By a direct calculation, one obtains the following:
\begin{align*}
& \mathscr{U}_\pm c(f) \mathscr{U}_\pm^*
= c((W_\pm)_{++} f) + c^*((W_\pm)_{+-} \bar{f}), \\
& \mathscr{U}_\pm c^*(f) \mathscr{U}_\pm^*
= c^*((W_\pm)_{++}f) + c((W_\pm)_{+-}\bar{f}), \\
& e^{i\psi(j_\pm)} c(f) e^{-i\psi(j_\pm)}
= c(f) -i(f,j_\pm), \\
& e^{i\psi(j_\pm)} c^*(f) e^{-i\psi(j_\pm)}
= c^*(f) +i(j_\pm,f).
\end{align*}
It holds from the above that
\begin{align*} 
& \mathscr{U}_\pm e^{i\psi(j_\pm)} 
			c^*\left( e^{it \omega_0} f \right) \Omega \\
& = \left( c^*\left( (W_\pm)_{++} e^{it\omega_0} f\right)
		+ c\left((W_\pm)_{+-} e^{-it\omega_0} \bar{f} \right) 
		+ i(g_\pm, e^{it\omega_0}f) \right) \mathscr{U}_\pm e^{i\psi(j_\pm)} \Omega \\
& = c^*\left( (W_\pm)_{++} e^{it\omega_0} f\right) 
	\mathscr{U}_\pm e^{i\psi(j_\pm)} \Omega + o(1)
\end{align*}
as $t \to \pm \infty$ in the strong topology
since $(W_\pm)_{+-}$ is Hilbert-Schmidt and
since $\lim_{t \to \pm \infty}(g_\pm, e^{it\omega_0}f)=0$
by the Riemann-Lebesgue Lemma.
Hence we have
\begin{align}
\label{0843}
& \mathscr{I}_t(f,g) \notag \\
& = \Big( c^*\left( (W_+)_{++} e^{it\omega_0} f\right) 
		\mathscr{U}_+ e^{i\psi(j_+)} \Omega , 
		c^*\left( (W_-)_{++} e^{it\omega_0} g\right) 
			\mathscr{U}_- e^{i\psi(j_-)} \Omega \Big)
+ o(1) \notag \\
& = \left((W_+)_{++} e^{it\omega_0} f, (W_-)_{++} e^{it\omega_0} g\right) 
	(\Omega, \mathscr{S}\Omega) \notag \\
& \quad +  \Big( c\left( (W_-)_{++} e^{it\omega_0} g\right) 
		\mathscr{U}_+ e^{i\psi(j_+)} \Omega , 
		c\left( (W_+)_{++} e^{it\omega_0} f\right) 
			\mathscr{U}_- e^{i\psi(j_-)} \Omega \Big) + o(1) 
\end{align}
as $t \to \pm \infty$.
It is straightforward that
\begin{align} 
& (W_+)_{++}^*(W_-)_{++} \notag \\
& = \frac{1}{4}\left(\omega_0^{-1/2} w_+^* \omega_0^{1/2}
	+ \omega_0^{+1/2} w_+^* \omega_0^{-1/2} \right)
	\left(\omega_0^{-1/2} w_- \omega_0^{1/2}
	+ \omega_0^{+1/2} w_- \omega_0^{-1/2} \right) \notag \\
\label{eq1426}
& = S + \frac{1}{4}w_+^* \left[(\omega^{-1/2}\omega_0 \omega^{-1/2} - 1)
 + (\omega^{1/2}\omega_0^{-1}\omega^{1/2} -1) \right]w_-.
\end{align}
Note that $S e^{it\omega_0} = e^{it\omega_0} S$ and that
the second term in the r.h.s of \eqref{eq1426} is Hilbert-Schmidt
by Lemma \ref{HSlemma}.
Hence, by \eqref{0843},
we have \eqref{0845}  
if the following holds true: 
\begin{equation} 
\label{0849}
\Big( c\left( (W_-)_{++} e^{it\omega_0} g\right) 
		\mathscr{U}_+ e^{i\psi(j_+)} \Omega , 
		c\left( (W_+)_{++} e^{it\omega_0} f\right) 
			\mathscr{U}_- e^{i\psi(j_-)} \Omega \Big) = o(1). 
\end{equation}
To prove \eqref{0849},
we use the relation
\begin{align*}
\mathscr{U}_\pm^* c(f) \mathscr{U}_\pm 
= c\left((W_\pm)_{++}^* \right) - c^*((W_\pm)_{-+}^* \bar{f}), 
\end{align*}
which is obtained from
$\mathscr{U}_\pm^* e^{i\psi(v)} \mathscr{U}_\pm
= e^{i\psi(QW_\pm^*Qv)}$.
By the above and the fact that 
$(W_\pm)_{-+}^*$ is Hilbert-Schmidt, 
we observe  that 
\begin{align*} 
c\left( (W_-)_{++} e^{it\omega_0} g\right) 
		\mathscr{U}_+ e^{i\psi(j_+)} \Omega
= \mathscr{U}_+ e^{i\psi(j_+)} 
	c\left( (W_-)_{++} e^{it\omega_0} g\right) \Omega + o(1).
\end{align*}
Since the first term of the above equals zero,
the proof of the lemma is completed.
\subsection{Characterization of the external source $J$}
Our aim of this subsection is 
to represent the classical source $J$ 
in terms of 
the functional
\[ F(t,f) = (\Omega_{\rm in}, \phi_{\rm out}(t, f) \Omega_{\rm in}), \quad f \in \mathcal{S}(\mathbb{R}^d). \]
We set
\begin{align*}
Z_\pm[f] &=X[f] \mp iY[f],
\end{align*} 
where
\[ X[f] = \left. \frac{d}{dt}F(t,\omega_0^{-1/2} f) \right|_{t=0}
\quad \mbox{and} \quad 
Y[f] = F(0,\omega_0^{+1/2} f). \]
Note that
\[ \phi_0(t, f) 
= \psi\left( v_t \right), \] 
where $v_t$ is 
\[ v_t = \begin{bmatrix} e^{it\omega_0} \omega_0^{-1/2} \bar f/\sqrt{2} \\ 
				e^{-it\omega_0} \omega_0^{-1/2} \bar f /\sqrt{2} \end{bmatrix}. \]
We see that
\[ C v_t = \begin{bmatrix} e^{it\omega_0} \omega_0^{-1/2} f/\sqrt{2} \\ 
				e^{-it\omega_0} \omega_0^{-1/2} f /\sqrt{2} \end{bmatrix} \not= v_t \]
and $v_t \not\in \mathcal{H}_C$. 
By a direct calculation, we have for $v \in \mathcal{H}_C$:
\begin{align*}
\mathscr{S}^{-1} \psi(v) \mathscr{S}
& = \psi(Q W_-^* Q W_+ v) + i (Qj_+ -W_+^* Q W_- j_-, v).
\end{align*} 
Since $(\Omega, \psi(v)\Omega) = 0$ and 
\begin{align*} 
(Qj_+ -W_+^* Q W_- j_-, v)
& = (W_+^* Q (W_+ Q j_+ - W_- j_-), v) \\
& = ( W_+ j_+ - W_- j_-, Q W_+ v),  
\end{align*}
one obtains
\[ (\Omega, \mathscr{S}^{-1} \psi(v) \mathscr{S} \Omega)  
= i  ( W_+ j_+ - W_- j_-, Q W_+ v),
\quad v \in \mathcal{H}_C. \]
It holds that
\begin{align*}
F(t,f) 
& = \frac{1}{2} (\Omega, \mathscr{S}^{-1} \psi(v_t + C v_t) \mathscr{S} \Omega) \\
& \quad + \frac{i}{2} (\Omega, \mathscr{S}^{-1} \psi(i (v_t - C v_t)) \mathscr{S} \Omega) \\
& = \frac{i}{2} (W_+ j_+ -W_- j_-, QW_+ (v_t + Cv_t)) \\
& \quad - \frac{i}{2} (W_+ j_+ -Q _- j_-, QW_+ (v_t - Cv_t)) \\
& = i (W_+ j_+ -W_- j_-, Q W_+ Cv_t). 
\end{align*}
Thus we infer
\begin{align*}
& X[f] = \frac{1}{\sqrt{2}} \left( W_+ j_+ - W_- j_-, QW_+ \begin{bmatrix} f \\ 
				- f \end{bmatrix} \right), \\
& Y[f] = \frac{i}{\sqrt{2}} \left(W_+ j_+ - W_- j_-, QW_+ \begin{bmatrix} f \\ 
				f \end{bmatrix} \right) 
\end{align*}
and
\begin{align*} 
& Z_+[f] 
= \sqrt{2}\left( W_+ j_+ - W_- j_-, QW_+ \begin{bmatrix} f \\
				0 \end{bmatrix} \right),\\
& Z_-[f] 
= -\sqrt{2}\left( W_+ j_+ - W_- j_-, QW_+ \begin{bmatrix} 0 \\
				f \end{bmatrix} \right).
\end{align*}
\begin{lemma}
Suppose that Assumptions \ref{asspot}, \ref{assscat}
and \ref{assext} hold.
Let
\begin{align*} 
Z[f,g] := Z_+[f]-Z_-[g]. 
\end{align*}
Then
\begin{align*}
& Z[Sf,  f]
= -2\sqrt{2} i ({\rm Im}(g_\infty), w_-f),
\\
& Z[S f, - f]
= 2\sqrt{2} ({\rm Re}(g_\infty), w_-f),
\end{align*}
where
\[ g_\infty := w_+ g_+ - w_- g_-
= - \frac{1}{\sqrt{2}} \int_{-\infty}^{+\infty} ds \omega^{-1/2} e^{is\omega} J_s. \]
\end{lemma}
\begin{proof}
By definition, we see that
\begin{align*} 
Z[w_+^*f, \pm w_-^*f] 
& =  \sqrt{2}\left( W_+ j_+ - W_- j_-, QW_+ 
	\begin{bmatrix} w_+^*f \\
				\pm w_-^*f \end{bmatrix} \right).
\end{align*}
By direct calculation, we have
\begin{align*} 
& W_+ j_+ -Q W_- j_- \\
& = 
\begin{bmatrix} 
\omega_0^{-1/2} \omega^{1/2} {\rm Re} (g_\infty)
	+ i \omega_0^{1/2} \omega^{-1/2} {\rm Im} (g_\infty) \\
\omega_0^{-1/2} \omega^{1/2} {\rm Re} (g_\infty)
	- i \omega_0^{1/2} \omega^{-1/2} {\rm Im} (g_\infty)
\end{bmatrix}
\end{align*}
and 
\begin{align*}
Q W_+
\begin{bmatrix}
	w_+^* f \\
	\pm w_-^* f
\end{bmatrix}
= 
\begin{bmatrix} 
	\omega_0^{\mp 1/2}\omega^{\pm 1/2}f \\ 
	\mp \omega_0^{\mp 1/2}\omega^{\pm 1/2}f
\end{bmatrix},
\end{align*}
which completes the proof.
\end{proof}

We introduce the functional  $Z : L^2(\mathbb{R}^3) \to \mathbb{C}$ by 
\begin{equation*}
Z[f] := \frac{1}{2\sqrt{2}}\left( 
	Z[Sf, f] + Z[S f, -f] \right), \quad f \in  L^2(\mathbb{R}^3). 
\end{equation*}

For $f\in L^2(\mathbb{R}^3)$, 
$\lambda>0$ and  
$k,x\in\mathbb{R}^3$, 
we denote 
$e^{-ik \cdot x}f(\lambda x)$ by 
$f_{k}^\lambda(x)$. 
$\mathcal{F}_0$ stands for the Fourier transform: 
$\mathfrak{h} \ni f \mapsto (\mathcal{F}_0f) = \hat{f}$
and $\hat{f}(k) = (2\pi)^{-3/2}\int dk e^{-ik\cdot x}f(x)$.
The generalized Fourier transform $\mathcal{F}_{\pm}$
is defined by $\mathcal{F}_0 w_\pm^*$.
Let $\chi \in \mathfrak{h}$ such that $0 \leq \chi \leq 1$
and $\chi(x) = 1$ if $|x| \leq 1$ and $\chi(x) = 0$ if $|x| \geq 2$.
We introduce a function $z_\lambda$ by
\[ z_\lambda(k) := -\sqrt{2}(2\pi)^{-3/2} Z[\chi_k^\lambda]. \]
\begin{lemma}
\label{zlambda}
Suppose that Assumptions \ref{asspot}, \ref{assscat}
and \ref{assext} hold.
Then $z_\lambda \in \mathfrak{h}$ and
\begin{align*}
\lim_{\lambda \to 0} z_\lambda 
= \int_{-\infty}^\infty ds 
	\mathcal{F}_+ \left(e^{-is \omega}\omega^{-1/2}J_s \right)
\quad \mbox{in $\mathfrak{h}$}.
\end{align*}
\end{lemma}

\begin{proof}
Since, by the above lemma, we have 
\begin{equation}
\label{zf} Z[f] = (w_-^* g_\infty, f),
\end{equation}
it follows that 
\begin{align*}
z_\lambda(k) 
& = \left( \int_{-\infty}^\infty ds 
w_-^* e^{is\omega} \omega^{-1/2}J_s, \chi_k^\lambda \right) \\
& = (2\pi)^{-3/2} \int dx e^{-ik \cdot x} \chi(\lambda x) 
	\times \left[\int_{-\infty}^\infty ds w_+^* e^{is\omega} \omega^{-1/2} J_s \right](k).
\end{align*} 
Hence $z_\lambda$ converges in $\mathfrak{h}$ to
\[ \mathcal{F}_0 \left[\int_{-\infty}^\infty ds 
	w_+^* e^{-is \omega}\omega^{-1/2}J_s \right]
= \int_{-\infty}^\infty ds 
	\mathcal{F}_+ e^{-is \omega}\omega^{-1/2} J_s. \] 
\end{proof}
Let
\begin{equation}
\label{z}
z : = \lim_{\lambda \to 0} z_\lambda. 
\end{equation}
\begin{proposition}
\label{039}
Suppose that Assumptions \ref{asspot} - \ref{assscat} 
and \ref{assext} hold.
If $\mathscr{S}(V,J)=\mathscr{S}(V^\prime,J^\prime)$,
then
\[ \int_{-\infty}^\infty ds e^{-is \omega}J_s
	= \int_{-\infty}^\infty ds e^{-is \omega}J^\prime_s. \]
\end{proposition}
\begin{proof}
Let $z^\prime(k)$ be defined as $z(k)$ 
with replacing $\mathscr{S}(V,J)$ by $\mathscr{S}(V^\prime,J^\prime)$.
By Theorem \ref{main1}, we have $S(V)=S(V^\prime)$, $V=V^\prime$ 
and hence $z = z^\prime$.
By Lemma \ref{zlambda}, we obtain
\[ \int_{-\infty}^\infty ds e^{-is \omega}\omega^{-1/2}J_s
= \int_{-\infty}^\infty ds e^{-is \omega}\omega^{-1/2}J^\prime_s \] 
by  the unitarity of $\mathcal{F}_+$.
\end{proof}
Henceforth, we suppose that 
$J$ is expressed by 
\begin{align}\label{Assumption:3.1}
J(t,x)=j(t)\rho(x),
\end{align}  
where 
$j \in L^1(\mathbb{R})$ and
$\rho \in H^{-1/2}(\mathbb{R}^3)$.
In Subsection \ref{subsec:3.1} below, 
assuming that 
$j$ is a given function and that
$\hat j$ is analytic,
we will represent $\rho$ in terms of $z$ and $j$. 
In Subsection \ref{subsec:3.2} below, 
we next assume that 
$\rho$ is a given function.
We will show that $j$ is determined by $z$ and $\rho$
if $j$ satisfies the following:
\begin{align}\label{condition:3.2}
\text{For some } \delta>0,\quad 
e^{\delta |t|} j(t)\in L^1(\mathbb{R}_t). 
\end{align}

\subsection{Reconstruction of $\rho$}\label{subsec:3.1}
Let $j$ be a given function belonging to 
$L^1(\mathbb{R}_t)$.
Then we immediately obtain 
the reconstruction formula for determining $\rho$: 
\begin{theorem}\label{Thm:3.1}
Suppose that Assumptions \ref{asspot}, \ref{assscat}
and \ref{assext} hold.
Assume that 
$J\in\mathscr{S}(\mathbb{R}\times\mathbb{R}^3)$
satisfies (\ref{Assumption:3.1}) 
and that $\hat j $ 
is a nonzero, real analytic given function. 
If 
$\omega({k})\notin( \hat j)^{-1}(0)$, 
then $(\mathcal{F}_+ \rho )({k})$ 
is uniquely determined by
\begin{align}\label{reconstruction:3.1}
(\mathcal{F}_+ \rho )({k})
=
\frac{z(k)}{(\hat j) (\hat \omega_0 ({k}))}.
\end{align} 
\end{theorem}
\begin{remark}
Since $\hat j$ 
is real analytic function, 
$(\hat j)^{-1}(0)$ is 
a discrete set and hence countable.
Thus, (\ref{reconstruction:3.1}) holds 
for almost all ${k}\in\mathbb{R}^3$ 
and we have $\rho = \mathcal{F}_+(z/(\hat j(\hat\omega_0))$.
If the generalized eigenfunction $\psi_+(k,x)$ of $-\Delta + V$ exists,
then the inverse of the generalized Fourier transform $\mathcal{F}_+$ is given by
\[ (\mathcal{F}_+^{-1} f) (x)
= (2\pi)^{-3/2}\int_{\mathbb{R}^3} dk \psi_+(k,x) f(k), \]
where we denote $\underset{R \to +\infty}{\mathrm{l.i.m}}\int_{|k| \leq R}dk$
by $\int_{\mathbb{R}^3} dk$.
In this case, we have
\begin{align*}
\rho(x)
= (2\pi)^{-3/2} \int_{\mathbb{R}^3} dk
\frac{\psi_+(k,x) z(k)}{(\hat j)(\hat \omega_0({k}))}.
\end{align*} 
\end{remark}

\subsection{Reconstruction of $j$}\label{subsec:3.2}
In this subsection
we suppose that $\rho$ is a nonzero, given function belonging to 
$H^{-1/2}(\mathbb{R}^3)$.
The following lemma will help us to identify $j$:
\begin{lemma}\label{lem:3.2}
Suppose that $j$ satisfies (\ref{condition:3.2}). 
Then $\hat j$ is real analytic. 
Furthermore, 
the radius of convergence of the Taylor series of 
$\hat j$ is larger than $\delta$.
\end{lemma}
\begin{proof}

By the assumption (\ref{condition:3.2}),  
it follows that for any non negative integer $m$,
\begin{align*}
|t|^m j(t) \in L^1(\mathbb{R}_t).
\end{align*} 
Therefore, we have 
$\hat j\in C^{\infty}(\mathbb{R})$
and 
\begin{align*}
(\hat j)^{(m)}(\tau)
=
\frac{1}{\sqrt{2\pi}}
\int_{\mathbb{R}}
(-it)^m
e^{-i\tau t}
j(t)
dt 
,\quad
\tau\in\mathbb{R},
\end{align*} 
where $(\hat j)^{(m)}$ is the 
$m$-th order derivative of $\hat j$.

Since 
\begin{align*}
\sup_{t>0}t^me^{-\delta t}
=
m^m \delta^{-m} e^{-m}, 
\end{align*} 
we obtain for any $\tau\in\mathbb{R}$, 
\begin{align*}
\left|
(\hat j)^{(m)}(\tau)
\right| 
&\le 
\frac{1}{\sqrt{2\pi}}
\int_\mathbb{R} |t|^m |j(t)|dt\\
&\le 
\frac{1}{\sqrt{2\pi}}
\int_\mathbb{R}
|j(t)|e^{\delta |t|}
|t|^m |e^{-\delta |t|}dt\\
&\le 
\frac{1}{\sqrt{2\pi}}
m^m \delta^{-m} e^{-m}
\int_\mathbb{R}
|j(t)|e^{\delta |t|}dt.
\end{align*} 
Thus, we have 
\begin{align*}
\left|
\frac{
(\widehat j)^{(m)}(\tau)}{
m!
}
\right| 
&\le 
\frac{1}{\sqrt{2\pi}}
\left( \int_\mathbb{R}
|j(t)|e^{\delta |t|}dt\right)
\frac{m^m  e^{-m}}{m!}\delta^{-m}.
\end{align*} 
Using Stirling's formula 
\begin{align*}
m!=\sqrt{2\pi}m^{m+1/2} e^{-m}e^{\theta(m)/12m}, 
\quad
0<\theta(m)<1,
\end{align*} 
we see that 
\begin{align*}
\limsup_{m\to \infty}\left|
\frac{
(\hat j)^{(m)}(\tau)}{
m!
}
\right|^{1/m}
\le \delta^{-1},
\end{align*} 
which completes the lemma.
\end{proof}
Applying Lemmas \ref{zlambda} and \ref{lem:3.2}, 
we have the following result:
\begin{theorem}
\label{Thm:3.2}
Suppose that Assumptions \ref{asspot}, \ref{assscat}
and \ref{assext} hold.
Assume that 
$J$ 
satisfies (\ref{Assumption:3.1}) 
and $\rho \in H^{-1/2}(\mathbb{R}^3)$ 
is a nonzero, given function.
If $j$ satisfies (\ref{condition:3.2}), 
then $j$ is uniquely reconstructed 
by the following steps:
\begin{enumerate}[(Step I)]
  \item Fix a point 
${k}_0\notin (\mathcal{F}_+ \rho)^{-1}(0)$. 
Let $U_0$ be a ${k}_0$-neighborhood 
such that $0\notin(\widehat J_X)(U_0)$.
Then we have 
\begin{align}\label{reconstruction:3.2}
(\hat j)(\omega({k}))
=
\frac{z(k)}{(\mathcal{F}_+\rho)(k)},
\quad{k}\in U_0.
\end{align} 
Therefore, we see exact values of 
$(\hat j)^{(m)}(\tau_0)$, 
$m=0,1,2,\cdots$,
where $\tau_0=\omega({k}_0)$.  
  \item 
For any $\tau\in (\tau_0-\delta,\tau_0+\delta)$, 
we have
\begin{align*}
(\hat j)(\tau)
=
\sum_{m=0}^\infty
\frac{(\hat j)^{(m)}(\tau_0)}{m!}
(\tau-\tau_0)^m.
\end{align*}    
  \item 
Let $l$ be a positive integer.
Suppose that we have already determined 
$(\hat j)(\tau)$
with 
$\tau\in (\tau_0 - \frac{(l+1)\delta}{2}, 
\tau_0 + \frac{(l+1)\delta}{2})$.  
For any  
$\tau\in [\tau_0+\frac{(l+1)\delta}{2},
\tau_0+\frac{(l+2)\delta}{2})$, 
we see the value of $(\widehat J_T)(\tau)$ by 
\begin{align*}
(\hat j)(\tau)
=
\sum_{m=0}^\infty
\frac{(\hat j)^{(m)}(\tau_0+ l\delta/2 )}{m!}
(\tau-\tau_0- l\delta/2)^m.
\end{align*}
On the other hand, 
For any  
$\tau\in (\tau_0-\frac{(l+2)\delta}{2},
\tau_0-\frac{(l+1)\delta}{2}]$, 
we see the value of $(\widehat J_T)(\tau)$ by 
\begin{align*}
(\hat j)(\tau)
=
\sum_{m=0}^\infty
\frac{(\hat j)^{(m)}(\tau_0- l\delta/2 )}{m!}
(\tau-\tau_0+ l\delta/2)^m.
\end{align*}
  \item
From (Step III), 
$\hat j$ is reconstructed completely. 
Hence we can determine $j$ 
by the inverse Fourier transform.  
\end{enumerate}
\end{theorem}

\appendix 
\section{Appendix}
\subsection{Hilbert-Schmidt operators}
We prove Lemma \ref{HSlemma}.
We first show that $\omega^{1/2} \omega_0^{-1/2} -1$.
To this end, we use
the formula 
\begin{equation}
\label{formula} 
A^{\alpha} = \frac{\sin (\pi \alpha)}{\pi} \int_0^\infty d\lambda \lambda^{\alpha-1} (A+\lambda)^{-1} A 
\end{equation}
on $D(A)$ ($0< \alpha < 1$).
With the aid of \eqref{formula} for $\alpha=1/4$, 
we have
\begin{align*} 
\omega^{1/2}\omega_0^{-1/2}-1 
& = [ (\omega^2)^{1/4}- (\omega_0^2)^{1/4} ] \omega_0^{-1/2} \\
& = \frac{1}{\sqrt{2}\pi} \int_0^\infty d\lambda 
	\lambda^{-3/4} [ (\omega^2+\lambda)^{-1}\omega^2 - (\omega_0^2+\lambda)^{-1}\omega_0^2] \omega_0^{-1/2} \\
& = \frac{1}{\sqrt{2}\pi} \int_0^\infty d\lambda 
	\lambda^{1/4} [ -(\omega^2+\lambda)^{-1} + (\omega_0^2+\lambda)^{-1}] \omega_0^{-1/2} \\
& = \frac{1}{\sqrt{2}\pi} \int_0^\infty d\lambda 
	\lambda^{1/4} (\omega^2+\lambda)^{-1} V (\omega_0^2+\lambda)^{-1} \omega_0^{-1/2},
\end{align*}
where the above equation holds true on some dense domain, e.g. $D(\omega_0^{3/2})$.
It suffices to prove that the operator $(\omega^2+\lambda)^{-1} V (\omega_0^2+\lambda)^{-1} \omega_0^{-1/2}$
is Hilbert-Schmidt and satisfy 
\begin{equation}
\label{intg1} 
\int_0^\infty d\lambda \lambda^{1/4} 
\|(\omega^2+\lambda)^{-1} V (\omega_0^2+\lambda)^{-1} \omega_0^{-1/2}\|_2 < \infty, 
\end{equation}
where $\| \cdot \|_2$ is the Hilbert-Schmidt norm. 
Since $V \in L^2(\mathbb{R}^3)$ and $(\sqrt{|k|^2+m^2})^{-3/2-\epsilon} \in L^2(\mathbb{R}^3)$,
$V \omega_0^{-3/2-\epsilon}$ is Hilbert-Schmidt for $\epsilon > 0$
(see \cite[Theorem XI.20]{RS3} for details).
Since $\|(\omega^2+\lambda)^{-1}\| \leq (m^2 + \lambda)^{-1}$,
we have
\begin{align*} 
& \|(\omega^2 + \lambda)^{-1}V(\omega_0^2+\lambda)^{-1}\omega^{-1/2}_0\|_2 \\
& \leq (m^2+\lambda)^{-1} \|V \omega_0^{-3/2-\epsilon}\|_2 \cdot \|\omega_0^{1+\epsilon}(\omega_0^2+\lambda)^{-1}\| \\
& \leq (m^2+\lambda)^{-1} \|V \omega_0^{-3/2-\epsilon}\|_2 \cdot
	\| \omega_0^{1+\epsilon}(\omega_0^2+\lambda)^{-(1+\epsilon)/2} \| 
	\cdot \|(\omega_0^2+\lambda)^{-1+(1+\epsilon)/2}\| \\
& \leq (m^2+\lambda)^{-3/2 + \epsilon/2}
\end{align*}
Thus \eqref{intg1} holds if $0 < \epsilon < 1/2$ and hence $\omega^{1/2}\omega_0^{-1/2}-1$ is Hilbert-Schmidt.

We shall prove that $\omega^{-1/2}\omega_0^{1/2} -1$ is Hilbert-Schmidt.
From a similar argument as above,
we infer that it suffices to show that
\begin{equation}
\label{intg2} 
\int_0^\infty \lambda^{1/4} \| \omega^{-1/2}(\omega^2+\lambda)^{-1}V(\omega_0^2+\lambda)^{-1}\|_2 < \infty.
\end{equation}
Since, as will be seen later, 
$\omega_0^{-1/2}V\omega_0^{-1-\epsilon}$ ($\epsilon>0$) is Hilbert-Schmidt,
we have
\[ \| \omega^{-1/2}(\omega^2+\lambda)^{-1}V(\omega_0^2+\lambda)^{-1}\|_2
\leq  C (m^2+\lambda)^{-3/2+\epsilon/2} \| \omega_0^{-1/2}V\omega_0^{-1-\epsilon} \|_2 \]
with some positive $C$.
Thus \eqref{integrable} holds true if $0 < \epsilon < 1/2$.

To prove that $\omega_0^{-1/2} (\omega_0 - \omega) \omega_0^{-1/2}$ is Hilbert-Schmidt,
we use the formula \eqref{formula} for $\alpha=1/2$ and write 
\[ \omega_0^{-1/2} (\omega_0 - \omega) \omega_0^{-1/2}
= -\frac{1}{\pi}\int_0^\infty d\lambda \lambda^{1/2} 
\omega_0^{-1/2} (\omega^2 + \lambda)^{-1} V (\omega_0^2 +\lambda)^{-1}\omega_0^{-1/2}. \]
where the above equation hold true on some dense domain (for instance $D(\omega_0^{3/2})$).
It suffices to show
\begin{equation}
\label{integrable} 
\int_0^\infty d\lambda \lambda^{1/2}
\|\omega_0^{-1/2} (\omega^2 + \lambda)^{-1} V (\omega_0^2 +\lambda)^{-1}\omega_0^{-1/2}\|_2 < \infty.
\end{equation}
If $\omega_0^{-1/2}V\omega_0^{-1-\epsilon}$ ($\epsilon>0$) is Hilbert-Schmidt,
then
\begin{align*}
& \|\omega_0^{-1/2} (\omega^2 + \lambda)^{-1} V (\omega_0^2 +\lambda)^{-1}\omega_0^{-1/2}\|_2 \\
& \leq \| \omega_0^{-1/2}(\omega^2 +\lambda)^{-1}\omega_0^{1/2}\|
\cdot \| \omega_0^{-1/2}V\omega_0^{-1-\epsilon}\|_2 
\cdot \| \omega_0^{1/2 + \epsilon}(\omega_0^2+\lambda)^{-1}\|  \\
& \leq \|\omega_0^{-1/2}\omega^{1/2}\| \cdot \|(\omega^2+\lambda)^{-1}\| \cdot \|\omega^{-1/2}\omega_0^{1/2}\| \\
& \quad \times \| \omega_0^{-1/2}V\omega_0^{-1-\epsilon}\|_2 
\cdot \|\omega^{1/2+\epsilon}(\omega_0^2+\lambda)^{-(1/2+\epsilon)/2}\| \cdot \|(\omega_0^2+\lambda)^{-3/4 + \epsilon/2}\| \\
& \leq C (m^2 + \lambda)^{\epsilon/2-7/4}
\end{align*}
with some $C>0$ and $0 < \epsilon < 3/2$.
This implies that \eqref{intg2} holds true for $0 < \epsilon < 1/2$
and we obtain the desired results.

It remains to show the following.
\begin{lemma}
For any $\epsilon>0$, $\omega_0^{-1/2}V\omega_0^{-1-\epsilon}$ is Hilbert-Schmidt.
\end{lemma}
\begin{proof}
Note that
\[ \omega_0^{-1/2} V \omega_0^{-1-\epsilon}
= V \omega_0^{-3/2-\epsilon} + [\omega_0^{-1/2},V]\omega_0^{-1-\epsilon}. \]
As seen above, the first term of the r.h.s. is Hilbert-Schmidt. 
It suffices to show that $[\omega_0^{-1/2},V]$ is Hilbert-Schmidt.
To this end, we write it as
\[ [\omega_0^{-1/2},V] 
= \frac{1}{\pi} \int_0^\infty d\lambda \lambda^{-1/2} [(\omega_0^2 + \lambda)^{-1},V]. \] 
By direct calculation, we see that
\begin{align*} 
[(\omega_0^2 + \lambda)^{-1},V] 
& = (\omega_0^2 + \lambda)^{-1}[V, \omega_0^2](\omega_0^2 + \lambda)^{-1} \\
& = V_1 + V_2,
\end{align*}
where 
\begin{align*} 
& V_1 = (\omega_0^2 + \lambda)^{-1}(\Delta V)(\omega_0^2 + \lambda)^{-1}, \\
& V_2 = 2 (\omega_0^2 + \lambda)^{-1}(\nabla V) \cdot \nabla (\omega_0^2 + \lambda)^{-1}. 
\end{align*}  
One observes that the operators $(\Delta V) (\omega_0^2 + \lambda)^{-1}$
and $(\omega_0^2 + \lambda)^{-1}(\nabla_j V)$ are integral operators with the kernels
$(4\pi|x-y|)^{-1}(\Delta V)(x)e^{\sqrt{m^2 + \lambda}|x-y|}$
and $(4\pi|x-y|)^{-1}e^{\sqrt{m^2 + \lambda}|x-y|}(\nabla_j V)(y)$,
respectively.
Hence we have 
\begin{align*}
\|(\Delta V) (\omega_0^2 + \lambda)^{-1}\|_2^2
& = (4\pi)^{-2}\int dx dy \frac{|(\Delta V)(x)|^2e^{-2\sqrt{m^2 + \lambda}|x-y|}}{|x-y|^2} \\
& = (4\pi)^{-2} \|(\Delta V)\|_{L^2}^2 \int dx \frac{e^{-2\sqrt{m^2 + \lambda}|x|}}{|x|^2} \\
& = \frac{1}{\sqrt{m^2 + \lambda}} \left(\frac{\|(\Delta V)\|_{L^2}^2}{(4\pi)^2} \int dx \frac{e^{-2|x|}}{|x|^2} \right)
\end{align*}
and 
\begin{align*}
\|(\omega_0^2 + \lambda)^{-1}(\nabla_j V) \|_2^2 
= \frac{1}{\sqrt{m^2 + \lambda}} \left(\frac{\|(\nabla V)\|_{L^2}^2}{(4\pi)^2} \int dx \frac{e^{-2|x|}}{|x|^2} \right). 
\end{align*}
Thus we obtain
\begin{align*}
\|V_1\|_2 
& \leq \|(\omega_0^2 + \lambda)^{-1} \| \cdot \|V (\omega_0^2 + \lambda)^{-1} \|_2 \\
& \leq C (m^2 + \lambda)^{-5/4} 
\end{align*}
and
\begin{align*}
\|V_2\|_2 
& \leq \sum_{j=1}^3 \|(\omega_0^2 + \lambda)^{-1}(\nabla_j V)\|_2
	 \cdot \|\nabla_j (\omega_0^2 + \lambda)^{-1/2}\| \cdot \|(\omega_0^2+\lambda)^{-1/2}\| \\
& \leq C (m^2 + \lambda)^{-3/4} 
\end{align*}
with some $C$ independent of $\lambda$.
Hence $[\omega_0^{-1/2}, V]$ is Hilbert-Schmidt since
\[ \|[\omega_0^{-1/2}, V]\|_2 
\leq C \int_0^\infty d\lambda \lambda^{-1/2} ((m^2 + \lambda)^{-5/4} + (m^2 + \lambda)^{-3/4}) < \infty. \]
\end{proof}

\subsection{Existence of the limits \eqref{1738}}
Under Assumptions \ref{asspot} and \ref{assscat},
we prove the existence of 
\begin{align*}
\mbox{s-}\lim_{t \to \pm \infty} 
e^{it\omega}e^{-it\omega_0}.
\end{align*} 
Since $\omega-\omega_0$ is $\omega_0$-bounded and
$\mathcal{S}(\mathbb{R}^3)$ is dense in 
$L^2(\mathbb{R}^3)=\mathcal{H}_{\mathrm{ac}}(\omega_0)$,
it suffices to show that 
\begin{align}\label{cook}
\int_1^\infty
dt
\left\|
(\omega-\omega_0)e^{-it\omega_0}f
\right\|_{\mathfrak{h}}
< \infty
\end{align} 
for any $f\in \mathcal{S}(\mathbb{R}^3)$.
From (\ref{formula}),
we see that 
\begin{align*}
\omega - \omega_0
&=
\frac{1}{\pi}
\int_0^\infty d\lambda 
\lambda^{-1/2}
[
(\omega^2 + \lambda)^{-1} \omega^2
-
(\omega_0^2 + \lambda)^{-1} \omega_0^2
]   \\
&=
\frac{1}{\pi}
\int_0^\infty d\lambda 
\lambda^{-1/2}
[
1-\lambda(\omega^2 + \lambda)^{-1} 
-
1+\lambda(\omega_0^2 + \lambda)^{-1}
]   \\
&=
\frac{1}{\pi}
\int_0^\infty d\lambda 
\lambda^{1/2}
[
(\omega^2 + \lambda)^{-1} 
-
(\omega_0^2 + \lambda)^{-1}
]   \\
&=
\frac{1}{\pi}
\int_0^\infty d\lambda 
\lambda^{1/2}
(\omega^2 + \lambda)^{-1}
[
(\omega_0^2 + \lambda)
-
(\omega^2 + \lambda)
]   
(\omega_0^2 + \lambda)^{-1}
\\
&=
-
\frac{1}{\pi}
\int_0^\infty d\lambda 
\lambda^{1/2}
(\omega^2 + \lambda)^{-1}
V 
(\omega_0^2 + \lambda)^{-1}.
\end{align*} 
Therefore, we have for any $f\in \mathcal{S}(\mathbb{R}^3)$,
\begin{align*}
&\| 
(\omega -\omega_0)e^{-it\omega_0}f
\|_{\mathfrak{h}} \\
\quad&\le C
\int_0^\infty
d\lambda
\lambda^{1/2}
\left\|
(\omega^2 + \lambda)^{-1}
V 
(\omega_0^2 + \lambda)^{-1}
e^{-it\omega_0} f
\right\|_{\mathfrak{h}}\\
\quad&\le C
\int_0^\infty
d\lambda
\lambda^{1/2}
(m^2+\lambda)^{-1}
\left\|
V 
(\omega_0^2 + \lambda)^{-1}
e^{-it\omega_0} f
\right\|_{\mathfrak{h}}\\
\quad&\le C
\left\|
V 
\right\|_{L^2(\mathbb{R}^3)}
\int_0^\infty
d\lambda
\lambda^{1/2}
(m^2+\lambda)^{-1}
\left\| 
(\omega_0^2 + \lambda)^{-1}
e^{-it\omega_0} f
\right\|_{L^\infty(\mathbb{R}^3)}.
\end{align*} 
Here,
we have used the H\"older inequality
in the last inequality.
It follows from \cite{MSW1980} that 
\begin{align*}
\left\| 
(\omega_0^2 + \lambda)^{-1}
e^{-it\omega_0} f
\right\|_{L^\infty(\mathbb{R}^3)}
\le C
|t|^{-3/2}
\left\| 
(\omega_0^2 + \lambda)^{-1}
\omega_0^{5/2}
f
\right\|_{L^1(\mathbb{R}^3)}
\end{align*} 
for any $t\in \mathbb{R}$.
Since 
\begin{align*}
&\left\| 
(\omega_0^2 + \lambda)^{-1}
\omega_0^{5/2}
f
\right\|_{L^1(\mathbb{R}^3)} \\
\quad& \le
\left\| 
(1+|x|)^{-2}
(1+|x|)^2
(\omega_0^2 + \lambda)^{-1}
\omega_0^{5/2}
f
\right\|_{L^1(\mathbb{R}^3)} \\
\quad& \le
\left\| 
(1+|x|)^{-2}
\right\|_{L^2(\mathbb{R}^3)} 
\left\| 
(1+|x|)^2
(\omega_0^2 + \lambda)^{-1}
\omega_0^{5/2}
f
\right\|_{L^2(\mathbb{R}^3)} \\
\quad& \le C
\left\| 
(\omega_0^2 + \lambda)^{-1}
\omega_0^{5/2}
f
\right\|_{L^2(\mathbb{R}^3)}
+
C
\left\| 
|x|^2
(\omega_0^2 + \lambda)^{-1}
\omega_0^{5/2}
f
\right\|_{L^2(\mathbb{R}^3)} 
\end{align*} 
and 
\begin{align*}
|x|^2 (\omega_0^2 + \lambda)^{-1}
=
(\omega_0^2 + \lambda)^{-1} |x|^2
+
8(\omega_0^2 + \lambda)^{-3}\Delta
-
6(\omega_0^2 + \lambda)^{-2},
\end{align*} 
we see that
\begin{align*}
&\left\| 
(\omega_0^2 + \lambda)^{-1}
\omega_0^{5/2}
f
\right\|_{L^1(\mathbb{R}^3)} \\
\quad&\le
C (m^2 + \lambda)^{-1}
\left\{
\left\| 
\omega_0^{5/2}
f
\right\|_{\mathfrak{h}} 
+
\left\| 
|x|^2
\omega_0^{5/2}
f
\right\|_{\mathfrak{h}} 
+
\left\| 
\Delta
\omega_0^{5/2}
f
\right\|_{\mathfrak{h}} 
\right\}.
\end{align*} 
Thus, we obtain for any $t\in \mathbb{R}$, 
\begin{align*}
\left\|
(\omega-\omega_0)e^{-it\omega_0}f
\right\|_{\mathfrak{h}}
\le C
|t|^{-3/2}
\int_0^\infty
d\lambda
\lambda^{1/2}
(m^2+\lambda)^{-2}
\end{align*} 
and hence (\ref{cook}) holds
for any $f\in \mathcal{S}(\mathbb{R}^3)$.
\subsection{Classical wave operator}

We prove Lemma \ref{waveop}.
It holds that
\[ U(t)^*U_0(t)
= \begin{bmatrix} U_{++}(t)^*e^{-it\omega_0} & U_{-+}(t)^*e^{it\omega_0} \\ 
	U_{+-}(t)^*e^{-it\omega_0} & U_{--}(t)^*e^{it\omega_0} \end{bmatrix}. \]
For $a,b,c, t$ and $s \in \mathbb{R}$, we set
\[ I_{(a,b,c)}(s,t) 
= \omega_0^a e^{is\omega} \omega^b e^{it\omega_0} \omega_0^c. \]
If $a+b+c=0$, 
then $I_{(a,b,c)}(s,t)$ is bounded and
\[ \overline{I_{(a,b,c)}(s,t)} = I_{(a,b,c)}(-s,-t). \]
By direct calculation, we have
\begin{align*}
U_{++}(t)^*e^{-it\omega_0}
& = \overline{U_{--}(t)^*e^{it\omega_0}} \\
& = \frac{1}{4}
\Big[
I_{(1/2,0,-1/2)}(t,-t) + I_{(-1/2,0,1/2)}(t,-t) \\
& \quad + I_{(1/2,-1,1/2)}(t,-t) + I_{(-1/2,1,-1/2)}(t,-t) \\
& \quad + I_{(1/2,0,-1/2)}(-t,-t) + I_{(-1/2,0,1/2)}(-t,-t) \\
& \quad -I_{(1/2,-1,1/2)}(-t,-t) - I_{(-1/2,1,-1/2)}(-t,-t) 
\Big]
\end{align*}
and
\begin{align*}
U_{-+}(t)^*e^{it\omega_0}
& = \overline{U_{+-}(t)^*e^{-it\omega_0}} \\
& = \frac{1}{4}
\Big[
I_{(-1/2,0,1/2)}(t,t) -I_{(1/2,0,-1/2)}(t,t) \\
& \quad + I_{(1/2,-1,1/2)}(t,t) - I_{(-1/2,1,-1/2)}(t,t) \\
& \quad + I_{(-1/2,0,1/2)}(-t,t) - I_{(1/2,0,-1/2)}(-t,t) \\
& \quad - I_{(1/2,-1,1/2)}(-t,t) + I_{(-1/2,1,-1/2)}(-t,t) 
\Big].
\end{align*}
By the similar argument as in the proof of Lemma \ref{dua}
with the aid of Lemma \ref{HSlemma},
one can prove the following lemma:
\begin{lemma}
Suppose that Assumptions \ref{asspot} and \ref{assscat}.
Then:
\begin{align*}
& \lim_{t \to \pm \infty}I_{(1/2,0,-1/2)}(t,-t)
= \lim_{t \to \pm \infty}I_{(1/2,-1,1/2)}(t,-t)
= \omega_0^{1/2} w_\pm \omega_0^{-1/2}, \\
& \lim_{t \to \pm \infty}I_{(-1/2,0,1/2)}(t,-t)
= \lim_{t \to \pm \infty}I_{(-1/2,1,-1/2)}(t,-t)
= \omega_0^{-1/2} w_\pm \omega_0^{1/2}, \\
& \lim_{t \to \pm \infty}I_{(1/2,0,-1/2)}(t,t)
= \lim_{t \to \pm \infty}I_{(1/2,-1,1/2)}(t,t)
= 0, \\
& \lim_{t \to \pm \infty}I_{(-1/2,0,1/2)}(t,t)
= \lim_{t \to \pm \infty}I_{(-1/2,1,-1/2)}(t,t)
= 0.
\end{align*}
\end{lemma}
By the above lemma, we have
\begin{align*} 
& \lim_{t \to \pm \infty}U_{++}(t)^*e^{-it\omega_0}
= \frac{1}{2}
	\left[\omega_0^{-1/2} w_\pm \omega_0^{+1/2} 
		+ \omega_0^{+1/2} w_\pm \omega_0^{-1/2}\right], \\
& \lim_{t \to \pm \infty}U_{-+}(t)^*e^{it\omega_0}
= \frac{1}{2}
	\left[\omega_0^{-1/2} w_\mp \omega_0^{+1/2} 
		- \omega_0^{+1/2} w_\mp \omega_0^{-1/2}\right]. \\
\end{align*}
This completes the lemma.

\begin{flushleft}
{\bf Acknowledgments}\\
The authors would like to thank A. Arai and T. Miyao for useful comments.
\end{flushleft}


\begin{thebibliography}{99}
\bibitem{Ar}A. Arai,
	Ground state of the massless Nelson model without infrared cutoff
	in a non-Fock representation,
	\textit{Rev. Math. Phys.} {\bf 13} (2001), 1075--1094.
\bibitem{Bachelot}A. Bachelot, 
	Inverse scattering problem for the nonlinear Klein-Gordon equation,  
	\textit{in} ``Contributions to nonlinear partial differential equations'',  
    pp.\ 7--15, 
    \textit{Res. Notes in Math.} \textbf{89}, Pitman, Boston, MA, 1983.
\bibitem{EnssWeder}V. Enss and R. Weder,
	The geometrical approach to multidimensional inverse scattering,
	\textit{J. Math. Phys.} {\bf 36} (1995), 3902--3921. 
\bibitem{Faddeev}L. D. Faddeev, 
	The inverse problem in the quantum theory of scattering, 
    \textit{Uspehi Mat. Nauk.} \textbf{14} (1959), 57--119.
\bibitem{HT}
	E. M. Henley and W. Thirring,
	\textit{Elementary Quantum Field Theory},
	McGraw-Hill, New York (1962).
\bibitem{MSW1980}	
	B. Marshall, W. Strauss, S. Wainger. $L^p$ - $L^q$ 
	Estimates for the Klein-Gordon equation, 
	\textit{J. Math. Pures et Appl.} {\bf 59} (1980), 417--440.
\bibitem{Menzala76}G. P. Menzala, 
       On the inverse problem for three-dimensional potential scattering,
       \textit{J. Differential Equations} \textbf{20} (1976), 233--247.
   \bibitem{Menzala77}G. P. Menzala, 
       Inverse scattering for the Klein-Gordon equation,
       \textit{Bull. Amer. Math. Soc.} \textbf{83} (1977), 735--736. 
   \bibitem{Menzala77-2}G. P. Menzala, 
       On inverse scattering for the Klein-Gordon equation with small potentials, 
       \textit{Funkcial. Ekvac.} \textbf{20} (1977), 61--70. 
	\bibitem{Morawetz-Strauss}C. Morawetz and W. A. Strauss, 
       On a nonlinear scattering operator,
       \textit{Comm. Pure Appl. Math.} \textbf{26} (1973), 47--54.
	\bibitem{PS}
		M. E. Peskin and D. V. Schroeder,
		\textit{An introduction to quantum field theory},
		Westview Press, the United State of America (1995).
	\bibitem{RS3}
		M. Reed and B. Simon,
		\textit{Methods of Modern Mathematical Physics, Vol III}, 
		Academic Press, New York (1979).
	\bibitem{Ruijsenaars}
		S. N. M .Ruijsenaars,
		On Boroliubov transformations. II. The general cases,
		\textit{Ann. Phys.} {\bf 116} (1978), 105--134.
   \bibitem{Sasaki07}H. Sasaki, 
      The inverse scattering problem for Schr\"{o}dinger and 
      Klein-Gordon equations with a nonlocal nonlinearity, 
      \textit{Nonlinear Anal. Theory, Methods \& Applications} 
      \textbf{66} (2007), 1770--1781.
   \bibitem{Sasaki-Watanabe}H. Sasaki and M. Watanabe, 
      Uniqueness on identification of cubic convolution nonlinearity, 
      \textit{J. Math. Anal. Appl.} \textbf{309} (2005), 294--306.
	\bibitem{Weder00-1}R. Weder, 
      Inverse scattering on the line for the nonlinear 
      Klein-Gordon equation with a potential,
      \textit{J. Math. Anal. Appl.} \textbf{252} (2000), 102--123.
   \bibitem{Weder02}R. Weder, 
      Multidimensional inverse scattering for the nonlinear 
      Klein-Gordon equation with a potential,  
      \textit{J. Differential Equations} \textbf{184} (2002), 62--77.
	\bibitem{WightmanGarding}
		A. S. Wightman and L. G\aa rding,
		Field as operator-valued distributions in quantum field theory,
		\textit{Ark. Fys.} {\bf 28} (1964), 129--184.
	\bibitem{Woll}
		M. Wollenberg,
		The invariance principle for wave operators,
		\textit{Pacific J. Math.} {\bf 59} (1975), 303--315.
\end{thebibliography}
\end{document}